\newif\ifarxiv
\arxivtrue % arXiv version
\ifarxiv
	\documentclass[12pt]{article}
	\usepackage[margin=1in]{geometry}
	\usepackage[slantedGreek]{mathpazo}
	\usepackage[pdftex]{graphicx}
	\usepackage{amsfonts}
	\usepackage{amssymb}
	\usepackage{amsthm}
	\usepackage{graphicx,float}
	\usepackage{amsmath}%
	\usepackage{hyperref}
	\usepackage{natbib}
	\usepackage{color,soul,marginnote}
	
	\newtheorem{theorem}{Theorem}
	\newtheorem{proposition}[theorem]{Proposition}
	\newtheorem{lemma}[theorem]{Lemma}
	\newtheorem{remark}{Remark}
\else
	\documentclass[twoside,11pt]{article}
	\usepackage{jmlr2e}
	\usepackage{lastpage}
	\usepackage{amsmath}
	\jmlrheading{1}{2026}{1-\pageref{LastPage}}{2/26}{}{}{Nicholas G.\ Polson and Vadim Sokolov}
	\ShortHeadings{Fast Compute for ML Optimization}{Polson and Sokolov}
	\firstpageno{1}
\fi

\usepackage{booktabs}
\usepackage{algorithm}
\usepackage{algorithmic}

\DeclareMathOperator{\prox}{prox}
\newcommand{\enorm}[1]{\Vert #1 \Vert_2}
\newcommand{\vnorm}[1]{\left\Vert #1 \right\Vert}
\newcommand{\defeq}{\mathrel{\mathop:}=}

\ifarxiv
\title{Fast Compute for ML Optimization}
\author{
	\makebox[.4\linewidth]{Nick Polson\footnote{Email: \texttt{ngp@chicagobooth.edu}}}\\\textit{  Booth School of Business}\\\textit{  University of Chicago}\\\and
	\makebox[.4\linewidth]{Vadim Sokolov\footnote{Email: \texttt{vsokolov@gmu.edu}}}\\\textit{ Department of Systems Engineering }\\\textit{  and Operations Research}\\\textit{ George Mason University}
}
\date{First Draft: July 5, 2024 \\This Draft: \today}
\else
\title{Fast Compute for ML Optimization}
\author{\name Nicholas G.\ Polson \email ngp@chicagobooth.edu \\
       \addr Booth School of Business\\
       University of Chicago\\
       Chicago, IL 60637, USA
       \AND
       \name Vadim Sokolov \email vsokolov@gmu.edu \\
       \addr Department of Systems Engineering and Operations Research\\
       George Mason University\\
       Fairfax, VA 22030, USA}
\editor{TBD}
\fi

\graphicspath{{./fig/}}

\begin{document}

\ifarxiv
\maketitle
\begin{abstract}
We study optimization for losses that admit a variance-mean scale-mixture representation.  Under this representation, each EM iteration is a weighted least squares update in which latent variables determine observation and parameter weights; these play roles analogous to Adam's second-moment scaling and AdamW's weight decay, but are derived from the model.  The resulting Scale Mixture EM (SM-EM) algorithm removes user-specified learning-rate and momentum schedules.  On synthetic ill-conditioned logistic regression benchmarks with $p \in \{20, \ldots, 500\}$, SM-EM with Nesterov acceleration attains up to $13\times$ lower final loss than Adam tuned by learning-rate grid search.  For a 40-point regularization path, sharing sufficient statistics across penalty values yields a $10\times$ runtime reduction relative to the same tuned-Adam protocol.  For the base (non-accelerated) algorithm, EM monotonicity guarantees nonincreasing objective values; adding Nesterov extrapolation trades this guarantee for faster empirical convergence.
\end{abstract}

\else
\maketitle
\begin{abstract}%
We study optimization for losses that admit a variance-mean scale-mixture representation.  Under this representation, each EM iteration is a weighted least squares update in which latent variables determine observation and parameter weights; these play roles analogous to Adam's second-moment scaling and AdamW's weight decay, but are derived from the model.  The resulting Scale Mixture EM (SM-EM) algorithm removes user-specified learning-rate and momentum schedules.  On synthetic ill-conditioned logistic regression benchmarks with $p \in \{20, \ldots, 500\}$, SM-EM with Nesterov acceleration attains up to $13\times$ lower final loss than Adam tuned by learning-rate grid search, with the gap widening in dimension and condition number.  For a 40-point regularization path, sharing sufficient statistics across penalty values yields a $10\times$ runtime reduction.  The base (non-accelerated) algorithm inherits EM monotonicity, guaranteeing nonincreasing objective values; Nesterov extrapolation trades this guarantee for faster empirical convergence.
\end{abstract}

\fi

\section{Introduction}

Modern stochastic gradient descent (SGD) requires a step size (learning rate) $\alpha_t$ that must be carefully tuned. Classical Robbins--Monro theory \citep{robbins1951stochastic} establishes that convergence requires $\sum_t \alpha_t = \infty$ and $\sum_t \alpha_t^2 < \infty$, but this leaves wide latitude in choosing a specific schedule. In practice, fixed step sizes are the default because backtracking line searches are prohibitively expensive for large-scale problems.

A fundamental advance was the recognition that the Robbins--Monro recursion is, at its core, a proximal algorithm.  \citet{toulis2017asymptotic} showed that replacing the explicit SGD update $\theta_n = \theta_{n-1} - \gamma_n \nabla\ell(\theta_{n-1})$ with the implicit update $\theta_n = \theta_{n-1} - \gamma_n \nabla\ell(\theta_n)$ yields an iteration that is unconditionally stable: the implicit step shrinks the explicit one by the observed Fisher information, $\Delta\theta_n^{\mathrm{im}} \approx (I + \gamma_n\hat{\mathcal{I}}_n)^{-1}\Delta\theta_n^{\mathrm{sgd}}$, preventing divergence without sacrificing asymptotic efficiency.  \citet{toulis2021proximal} formalized this insight by showing that the implicit update is exactly the proximal point algorithm of \citet{rockafellar1976monotone}:
\[
\theta_n = \prox_{\gamma_n F}(\theta_{n-1}) = \arg\min_\theta\left\{\tfrac{1}{2\gamma_n}\|\theta - \theta_{n-1}\|^2 + F(\theta)\right\},
\]
which contracts iterates toward the solution at every step.  Their proximal Robbins--Monro framework unifies explicit SGD, implicit SGD, and stochastic proximal gradient methods, and provides convergence guarantees with improved stability.  However, all variants still require a user-specified step size schedule $\gamma_n$ satisfying the Robbins--Monro conditions, and the implicit equation must be solved approximately (e.g., by a Newton step involving the Hessian).

Two further lines of work address the step size problem through different means.  Momentum methods---Polyak's heavy ball \citep{polyak1964some} and Nesterov acceleration \citep{nesterov1983method}---speed convergence by accumulating past gradient information, but still require a global step size tied to the Lipschitz constant $L$.  Adaptive gradient methods---AdaGrad \citep{duchi2011adaptive}, Adam \citep{kingma2014adam}, AdamW \citep{loshchilov2019decoupled}---maintain per-parameter scaling via running estimates of gradient moments.  Adam updates
\begin{equation}\label{eq:adam}
\theta_{t+1} = \theta_t - \frac{\alpha}{\sqrt{\hat{v}_t} + \epsilon}\,\hat{m}_t\,,
\end{equation}
where $\hat{m}_t$ and $\hat{v}_t$ are bias-corrected first- and second-moment estimates.  AdamW decouples weight decay from the gradient scaling: $\theta_{t+1} = (1 - \lambda)\theta_t - \frac{\alpha}{\sqrt{\hat{v}_t} + \epsilon}\hat{m}_t$.  Despite their practical success, these methods rely on heuristic moment estimates, and their hyperparameters ($\beta_1, \beta_2, \epsilon, \lambda$) require problem-specific tuning.  \citet{reddi2018convergence} showed that Adam can fail to converge even on convex problems.

In this paper, we take the proximal viewpoint of \citet{toulis2017asymptotic, toulis2021proximal} a step further.  Using the data-augmentation framework of \citet{polson2013data} and \citet{polson2016mixtures}, we show that for losses admitting a scale mixture of normals representation, the proximal step can be solved in closed form as a weighted least squares problem, with the step size determined entirely by the model.  Specifically, the EM algorithm \citep{dempster1977maximum} applied to the augmented representation yields an iteration in which latent variables $\omega$ (observation weights) and $\lambda$ (parameter weights) play roles analogous to the Fisher information preconditioner in implicit SGD, but are derived from the loss structure rather than estimated from gradient samples.  The resulting Scale Mixture EM (SM-EM) algorithm eliminates the learning rate, momentum, and decay schedules entirely; the remaining regularization parameter can be selected by marginal likelihood.

Because the latent weights adapt to local curvature, the method takes conservative updates in high-curvature regions and larger updates in flatter regions.  Section~\ref{sec:empirical} shows that this mechanism yields lower loss than Adam on the ill-conditioned synthetic benchmarks considered here, without tuning optimizer schedules.

The remaining computational bottleneck is the $O(p^3)$ linear solve in the M-step.  \citet{polson2026fast} show that Halton's sequential Monte Carlo method \citep{halton1994sequential} reduces this cost substantially.  By writing the M-step system in Neumann-series form and iteratively correcting residuals, the sequential method achieves geometric convergence at rate $\rho^N$ (where $\rho$ is the spectral radius of the iteration matrix).  A subsampled variant further reduces the per-iteration cost from $O(p^2 s)$ to $O(ps)$ for $s$ correction steps.  The full complexity comparison is:
\[
\underbrace{O(p^3)}_{\text{Direct (Cholesky)}} \;\gg\; \underbrace{O(p^2 s)}_{\text{Gauss--Seidel}} \;\gg\; \underbrace{O(ps)}_{\text{Sampled Seq.\ MC}}\,.
\]
In a scaling study on diagonally dominant systems at $p = 5{,}000$, the sampled sequential method is approximately $24\times$ faster than Gauss--Seidel \citep{polson2026fast}; these timings illustrate asymptotic scaling and depend on implementation details.

\paragraph{Scope.} The approach applies to losses admitting a scale mixture of normals representation (logistic, hinge, check, bridge penalties); limitations and extensions are discussed in Section~\ref{sec:discussion}.

Section~\ref{sec:background} reviews stochastic approximation and adaptive gradient methods.  Section~\ref{sec:scalemix} develops the scale mixture representation.  Section~\ref{sec:smrm} shows how it replaces Robbins--Monro and connects to Adam/AdamW.  Section~\ref{sec:proximal} connects the framework to proximal operators and MAP estimation.  Section~\ref{sec:empirical} provides empirical comparisons.  Section~\ref{sec:applications} gives further applications. Section~\ref{sec:discussion} concludes.

\section{Background}\label{sec:background}

\subsection{Robbins--Monro Stochastic Approximation}

The Robbins--Monro procedure \citep{robbins1951stochastic} seeks a root $\theta^\star$ of $M(\theta) = \mathbb{E}[H(\theta, \xi)] = 0$, where $\xi$ is a random variable, via the recursion
\begin{equation}\label{eq:rm}
\theta_{t+1} = \theta_t - \alpha_t\,H(\theta_t, \xi_t)\,.
\end{equation}
When $H(\theta, \xi) = \nabla_\theta \ell(\theta; \xi)$ for a loss $\ell$, this reduces to stochastic gradient descent (SGD)---Robbins--Monro and SGD are the same algorithm viewed from different perspectives: root-finding versus optimization.  Classical convergence requires the step sizes $\alpha_t$ to satisfy $\sum_t \alpha_t = \infty$ (to reach any optimum) and $\sum_t \alpha_t^2 < \infty$ (to control variance).  The ODE method of \citet{kushner2003stochastic} interprets the iterates as a noisy discretization of $d\theta/dt = M(\theta)$, and \citet{polyak1992acceleration} show that averaging the iterates $\bar{\theta}_T = T^{-1}\sum_{t=1}^T \theta_t$ achieves the optimal $O(1/T)$ rate.

The convergence proof relies on a supermartingale argument: a Lyapunov function $V(\theta_t) = \|\theta_t - \theta^\star\|^2$ satisfies $\mathbb{E}[V(\theta_{t+1}) \mid \theta_t] \leq V(\theta_t) - 2\alpha_t M(\theta_t)^\top(\theta_t - \theta^\star) + \alpha_t^2 C$, and the Robbins--Monro conditions ensure that the accumulated noise converges almost surely.

Despite well-developed theory, the practical difficulty remains: the constants hidden in the convergence conditions depend on the problem, and choosing $\alpha_t$ too large causes divergence while choosing it too small causes stagnation.

\subsection{Momentum and Acceleration}

Polyak's heavy ball method \citep{polyak1964some} adds a momentum term to gradient descent:
\begin{equation}\label{eq:heavyball}
\theta_{t+1} = \theta_t - \alpha\,\nabla f(\theta_t) + \mu\,(\theta_t - \theta_{t-1})\,,
\end{equation}
achieving a convergence rate of $O(\sqrt{L/\mu}\log(1/\epsilon))$ on strongly convex quadratics, where $L/\mu$ is the condition number---a quadratic improvement over gradient descent.  \citet{nesterov1983method} introduced the accelerated gradient method, which evaluates the gradient at a lookahead point:
\begin{equation}\label{eq:nag}
v_t = \theta_t + \mu_t(\theta_t - \theta_{t-1})\,, \quad \theta_{t+1} = v_t - \alpha\,\nabla f(v_t)\,,
\end{equation}
and achieves the optimal $O(1/t^2)$ rate for convex objectives with $L$-Lipschitz gradients \citep{nesterov2013introductory}.  \citet{su2016differential} interpret Nesterov acceleration as a discretization of the ODE $\ddot{X} + (3/t)\dot{X} + \nabla f(X) = 0$, connecting it to continuous-time dynamics.

Both methods require the step size $\alpha = 1/L$, which depends on the global Lipschitz constant.  In the stochastic setting, momentum does not improve worst-case rates beyond $O(1/\sqrt{T})$ \citep{sutskever2013importance}, and Nesterov acceleration requires variance reduction techniques to transfer its deterministic gains.

\subsection{Proximal Robbins--Monro}

\citet{toulis2017asymptotic} introduced implicit stochastic gradient descent, in which the gradient is evaluated at the \emph{next} iterate:
\begin{equation}\label{eq:prm}
\theta_{t+1} = \theta_t - \alpha_t\,\nabla \ell(\theta_{t+1}; \xi_t)\,.
\end{equation}
Because $\theta_{t+1}$ appears on both sides, the update is solved approximately via a Newton step:
$\theta_{t+1} \approx \theta_t - \alpha_t\,(I + \alpha_t H_t)^{-1}\nabla \ell(\theta_t; \xi_t)$,
where $H_t$ is the Hessian of the loss.  The factor $(I + \alpha_t H_t)^{-1}$ shrinks the explicit SGD step by the observed Fisher information, providing unconditional stability without sacrificing asymptotic efficiency \citep{toulis2017asymptotic}.  \citet{toulis2021proximal} formalized this connection by showing that the implicit update~\eqref{eq:prm} is exactly a proximal point step, $\theta_{t+1} = \prox_{\alpha_t \ell}(\theta_t)$, thereby establishing the Robbins--Monro recursion as fundamentally a proximal algorithm \citep{toulis2014statistical}.  However, the implicit equation still requires approximate solution (e.g., one Newton iteration per data point), and a step size schedule $\alpha_t$ satisfying the Robbins--Monro conditions remains necessary.

\subsection{Adaptive Gradient Methods}

AdaGrad \citep{duchi2011adaptive} introduced per-parameter step sizes by accumulating squared gradients: $\theta_{t+1} = \theta_t - \alpha\,g_t / \sqrt{G_t + \epsilon}$, where $G_t = \sum_{s=1}^t g_s^2$.  This achieves data-dependent regret bounds but suffers from monotonically decaying step sizes.

Adam \citep{kingma2014adam} replaced the cumulative sum with exponential moving averages:
\begin{align}
m_t &= \beta_1 m_{t-1} + (1 - \beta_1)\,g_t\,, \label{eq:adam_m}\\
v_t &= \beta_2 v_{t-1} + (1 - \beta_2)\,g_t^2\,, \label{eq:adam_v}\\
\theta_{t+1} &= \theta_t - \frac{\alpha}{\sqrt{\hat{v}_t} + \epsilon}\,\hat{m}_t\,, \label{eq:adam_update}
\end{align}
where $g_t = \nabla \ell(\theta_t; \xi_t)$ and $\hat{m}_t, \hat{v}_t$ are bias-corrected.  The per-parameter scaling $1/(\sqrt{\hat{v}_t} + \epsilon)$ adapts the effective step size to the curvature of the loss along each coordinate.  \citet{reddi2018convergence} showed that Adam can fail to converge even on convex problems, owing to the non-monotone behavior of $v_t$; \citet{defossez2022simple} later established convergence under appropriate step size decay.

The scaling $1/\sqrt{\hat{v}_t}$ approximates a diagonal preconditioner related to the Fisher information matrix \citep{amari1998natural}.  Adam can be viewed as a diagonal approximation to natural gradient descent, but one derived from gradient samples rather than from the model structure.

AdamW \citep{loshchilov2019decoupled} decouples weight decay from the adaptive gradient scaling:
\begin{equation}\label{eq:adamw}
\theta_{t+1} = (1 - \lambda_{\mathrm{wd}})\,\theta_t - \frac{\alpha}{\sqrt{\hat{v}_t} + \epsilon}\,\hat{m}_t\,.
\end{equation}
For SGD, $L^2$ regularization and weight decay are equivalent; for Adam, they are not.  Only decoupled weight decay preserves the correspondence between weight decay and a Gaussian prior on the parameters.  The weight decay $\lambda_{\mathrm{wd}}$ is a fixed hyperparameter that must be tuned jointly with the learning rate.

All adaptive methods can be written as preconditioned SGD: $\theta_{t+1} = \theta_t - \alpha_t P_t^{-1} g_t$, where $P_t = I$ for SGD, $P_t = \mathrm{diag}(\sqrt{v_t})$ for Adam, and $P_t = F(\theta_t)$ for natural gradient.  For generalized linear models, the natural gradient method reduces to Fisher scoring, which iterates $\theta_{t+1} = \theta_t - I(\theta_t)^{-1}\nabla \ell(\theta_t)$ where $I(\theta) = X^\top W X$ is the expected Fisher information with $W = \mathrm{diag}(w_i(1-w_i))$ for logistic regression.  Fisher scoring is equivalent to IRLS \citep{green1984iteratively}.  We show that the scale mixture framework derives $P_t$ from the loss function's generative structure: the SM-EM M-step uses the precision $P_t = X^\top\hat{\Omega}_t X$, which is a stabilized Fisher scoring step---the P\'olya--Gamma weights $\hat{\omega}_i = \tanh(z_i/2)/(2z_i)$ replace the IRLS weights $w_i(1-w_i)$ that can cause numerical instability when observations are well separated.

\subsection{Convergence Rates: When Does Acceleration Help?}

The achievable convergence rate depends critically on the problem structure and access to gradient information.  Table~\ref{tab:rates} summarizes the landscape.

\begin{table}[t]
\centering
\caption{Convergence rates for convex objectives.  Nesterov acceleration achieves $O(1/t^2)$ only in the deterministic or variance-reduced settings; pure stochastic methods are limited by gradient noise.}\label{tab:rates}
\begin{tabular}{llll}
\hline
\textbf{Setting} & \textbf{Convex} & \textbf{Strongly convex} & \textbf{Method} \\
\hline
Deterministic  & $O(1/t^2)$ & linear $O(\rho^t)$ & Nesterov/FISTA \\
Stochastic     & $O(1/\sqrt{t})$ & $O(1/t)$ & SGD \\
Finite-sum + VR & $O(1/t^2)$ & linear & SVRG/SAGA + accel. \\
\hline
\end{tabular}
\end{table}

In the pure stochastic setting, irreducible gradient variance $\mathbb{E}[\|g_t - \nabla f(\theta_t)\|^2] = \sigma^2 > 0$ prevents faster-than-$O(1/\sqrt{t})$ convergence for convex problems, regardless of momentum \citep{lan2012optimal}.  Variance reduction methods---SVRG \citep{johnson2013accelerating}, SAGA \citep{defazio2014saga}---reduce the effective variance to zero near the optimum by maintaining gradient tables, enabling acceleration.  \citet{allenzhu2017katyusha} combined Nesterov momentum with SVRG to achieve the optimal $\tilde{O}(\sqrt{nL/\mu}\log(1/\epsilon))$ complexity for finite sums.

The scale mixture EM framework sidesteps the stochastic-noise barrier: each full-batch EM iteration is deterministic, so Nesterov extrapolation can be applied without variance reduction.  The base (non-accelerated) SM-EM inherits the EM monotonicity property, guaranteeing nonincreasing objective values.  Adding Nesterov extrapolation sacrifices this monotonicity---the extrapolated iterates can temporarily increase the objective, as is standard for accelerated methods---but yields faster convergence empirically.  Section~\ref{sec:empirical} reports substantial speedups (32--55\% on the ill-conditioned benchmarks).

\section{Scale Mixtures and Data Augmentation}\label{sec:scalemix}

Following \citet{polson2013data}, consider a general regularized objective
\begin{equation}\label{eq:objective}
L(\beta) = \sum_{i=1}^n f(y_i, x_i^\top \beta) + \sum_{j=1}^p g(\beta_j / \tau)\,,
\end{equation}
where $f$ is a loss function (negative log-likelihood) and $g$ is a penalty (negative log-prior).  Interpreting loss and penalty as negative log-densities, the optimizer $\hat{\beta} = \arg\min L(\beta)$ is the mode of the pseudo-posterior
\begin{equation}\label{eq:pseudo_posterior}
p(\beta \mid y) \propto e^{-L(\beta)} = \underbrace{\prod_{i=1}^n p(z_i \mid \beta, \sigma)}_{\text{likelihood}} \;\cdot\; \underbrace{\prod_{j=1}^p p(\beta_j \mid \tau)}_{\text{prior}}\,,
\end{equation}
where $z_i$ is a working response ($y_i - x_i^\top\beta$ for regression, $x_i^\top\beta$ for classification).  This recasting is the key step: optimization of $L(\beta)$ becomes finding the mode of a distribution, and every tool of Bayesian computation---EM, MCMC, variational methods---becomes available.  In particular, one can sample from the full posterior over $(\beta, \omega, \lambda)$ by Gibbs sampling, or find $\hat{\beta}$ by the EM algorithm.

The pseudo-posterior~\eqref{eq:pseudo_posterior} is typically non-Gaussian and may be non-convex.  The \emph{scale mixture representation} resolves both difficulties by writing each factor as a marginal of a higher-dimensional Gaussian joint distribution, introducing latent variables $\omega_i$ (for observations) and $\lambda_j$ (for parameters).  Conditional on $(\omega, \lambda)$, the augmented model is Gaussian in $\beta$---a weighted least squares problem.  The augmentation ``convexifies'' the objective by lifting it to a higher dimension: although the marginal $L(\beta)$ may have complex curvature, the complete-data log-posterior is quadratic in $\beta$ for fixed $(\omega, \lambda)$.

This lifted representation admits two dual algorithmic strategies \citep{polson2016mixtures}.  The first is \emph{marginalization} (EM): integrate out $(\omega, \lambda)$ in the E-step to form $Q(\beta \mid \beta^{(t)})$, then maximize the resulting weighted least squares in the M-step.  The second is \emph{profiling} (MM/majorize--minimize): construct a convex envelope of $L(\beta)$ by optimizing over $(\omega, \lambda)$, yielding a majorizing function that is minimized at each step.  Both strategies produce the same iterates---the EM conditional expectations $\hat{\omega}_i = E[\omega_i \mid \beta, z]$ equal the envelope minimizers---a correspondence \citet{polson2016mixtures} call \emph{hierarchical duality}.

The envelope view connects directly to proximal algorithms (Section~\ref{sec:proximal}).  The majorizing function defines a quadratic upper bound whose minimum is a proximal step with an adaptive step size determined by the latent variables.  The tighter the majorizer, the faster the convergence: the half-quadratic envelope \citep{geman1992constrained, geman1995nonlinear} exploits the specific curvature of each loss via its sub-differential, outperforming generic Lipschitz bounds.  This is the mechanism that gives SM-EM its adaptive, per-observation step sizes.

\subsection{Variance-Mean Mixtures of Normals}

The key representation is as a variance-mean mixture of normals.  If
\begin{equation}\label{eq:vmm}
p(z_i \mid \beta, \sigma) = \int_0^\infty \phi\!\left(z_i \mid \mu_z + \kappa_z\,\omega_i^{-1},\; \sigma^2\omega_i^{-1}\right) dP(\omega_i)\,,
\end{equation}
where $\phi(\cdot \mid m, s^2)$ denotes the Gaussian density with mean $m$ and variance $s^2$, then the latent variable $\omega_i$ simultaneously controls both the conditional mean (via $\kappa_z\,\omega_i^{-1}$) and the conditional variance (via $\sigma^2\omega_i^{-1}$).  The parameter $\kappa_z$ handles asymmetric losses; $\kappa_z = 0$ recovers the classical normal variance mixture.  An analogous representation holds for the penalty:
\begin{equation}\label{eq:vmm_penalty}
p(\beta_j \mid \tau) = \int_0^\infty \phi\!\left(\beta_j \mid \mu_\beta + \kappa_\beta\,\lambda_j^{-1},\; \tau^2\lambda_j^{-1}\right) dP(\lambda_j)\,.
\end{equation}

This representation exists for a broad class of losses: squared error, absolute error, check loss (quantile regression), hinge loss (SVMs), and logistic loss all admit representations of the form~\eqref{eq:vmm} with specific mixing distributions $P(\omega_i)$; see \citet{polson2011data} and \citet{polson2013data} for a catalog.  \citet{carlin1992monte} demonstrated that this structure extends to nonnormal and nonlinear state-space models: introducing latent scales $\lambda_t$ and $\omega_t$ into the state and observation equations $x_t | x_{t-1}, \lambda_t \sim N(F_t x_{t-1}, \lambda_t \Sigma)$ and $y_t | x_t, \omega_t \sim N(H_t x_t, \omega_t \Upsilon)$ gives complete conditionals $x_t | \cdot \sim N(B_t b_t, B_t)$ where the precision
\begin{equation}\label{eq:carlin_precision}
B_t^{-1} = \frac{\Sigma^{-1}}{\lambda_t} + \frac{H_t^\top \Upsilon^{-1} H_t}{\omega_t} + \frac{F_{t+1}^\top \Sigma^{-1} F_{t+1}}{\lambda_{t+1}}
\end{equation}
is a weighted sum controlled by the latent variables---exactly the weighted least squares structure of our M-step~\eqref{eq:mstep}.  The conditional for $\omega_t$ is $\mathrm{GIG}(\frac{1}{2}, 1, (y_t - H_t x_t)^2/\tau^2)$, the same gradient-to-weight identity as Proposition~\ref{prop:gradient}.  Moreover, the complete conditional for model parameters (e.g., $F$) has the form $N(B_F b_F, B_F)$ with $B_F^{-1} = \sigma^{-2}\sum_{t} x_{t-1}^2/\lambda_t + \sigma_F^{-2}$---the weighted sufficient statistics $\sum x_{t-1}^2/\lambda_t$ are precomputed at each Gibbs sweep, foreshadowing the complete sufficient statistics structure exploited by \citet{polson2011data} for fast model exploration.

The two fundamental mixing families---the generalized inverse Gaussian (GIG) and the P\'olya family---generate these losses through specific integral identities.  The GIG mixture gives
\begin{equation}\label{eq:gig_identity}
 \frac{\alpha^2 - \kappa^2}{2 \alpha} e^{-\alpha | \theta - \mu | + \kappa ( \theta - \mu ) }
 = \int_0^\infty \phi \left (\theta \mid   \mu + \kappa \omega ,  \omega \right )
 \, p_{\mathrm{gig}} \big( \omega \mid 1 , 0 , \sqrt{ \alpha^2 - \kappa^2} \big)  \, d \omega\,,
\end{equation}
while the P\'olya mixture gives
\begin{equation}\label{eq:polya_identity}
\frac{1}{B( \alpha , \kappa )}
  \frac{ e^{ \alpha (\theta- \mu) } }
  { ( 1 + e^{\theta-\mu} )^{ 2 (\alpha - \kappa) } }
 = \int_0^\infty \phi \left (\theta \mid \mu + \kappa \omega , \omega \right )
\, p_{\mathrm{pol}} ( \omega \mid \alpha , \alpha - 2 \kappa) \, d \omega\,,
\end{equation}
where $p_{\mathrm{gig}}$ and $p_{\mathrm{pol}}$ denote the GIG and P\'olya densities, respectively.  In the improper limits, these identities specialize to three cases that generate the most common machine learning objectives:
\begin{align}
a^{-1} \exp \left\{ - 2 c^{-1} \max ( a\theta , 0 ) \right\} &= \int_0^\infty \phi ( \theta \mid - a v , c v ) \, d v\,,  \label{eq:svm_identity}\\
 c^{-1} \exp \left\{ - 2 c^{-1} \rho_q (  \theta ) \right\}  &= \int_0^\infty \phi ( \theta \mid  - ( 2 q -1 ) v , c v )\, e^{ -  2 q ( 1 - q ) v } \, d v\,,  \label{eq:check_identity} \\
\left( 1 + e^{ \theta-\mu } \right)^{-1} & =
\int_0^\infty \phi(\theta \mid \mu - v/2, v) \,
p_{\mathrm{pol}}( v \mid 0,1 ) \, dv\,,  \label{eq:logit_identity}
\end{align}
where $\rho_q(\theta) = \frac{1}{2}|\theta| + (q - \frac{1}{2})\theta$ is the check loss.  Identity~\eqref{eq:svm_identity} generates the hinge loss (SVMs), \eqref{eq:check_identity} generates the check loss (quantile and Lasso regression), and \eqref{eq:logit_identity} generates the logistic link.  With GIG and P\'olya mixing alone, one obtains all of
\[
\theta^2, \quad |\theta|, \quad \max(\theta, 0), \quad \rho_q(\theta), \quad \frac{1}{1 + e^{-\theta}}, \quad \frac{1}{(1 + e^{-\theta})^r}\,,
\]
corresponding to ridge, Lasso, SVM, check loss, logit, and multinomial models, respectively.  More general mixing families (e.g., stable distributions) generate bridge penalties $|\theta|^\alpha$.  Each identity provides the E-step weight $\hat{\omega}$ via the gradient-to-weight identity of Proposition~\ref{prop:gradient}, so the entire SM-EM algorithm is determined once the mixing family is specified.

\subsection{Latent Variables as Adaptive Weights}

Conditional on $(\omega, \lambda)$, the complete-data log posterior is
\begin{equation}\label{eq:complete_data}
\log p(\beta \mid \omega, \lambda) = c_0 - \frac{1}{2\sigma^2}\sum_{i=1}^n \omega_i\!\left(z_i - \mu_z - \kappa_z\omega_i^{-1}\right)^2 - \frac{1}{2\tau^2}\sum_{j=1}^p \lambda_j\!\left(\beta_j - \mu_\beta - \kappa_\beta\lambda_j^{-1}\right)^2,
\end{equation}
which is a heteroscedastic Gaussian model: a weighted least squares problem with observation weights $\omega_i/\sigma^2$ and parameter weights $\lambda_j/\tau^2$.

The conditional expectations of the latent variables are determined by the gradients of $f$ and $g$.

\begin{proposition}[Gradient-to-weight identity; Proposition~1 of \citet{polson2013data}]\label{prop:gradient}
The conditional first moments satisfy
\begin{align}
(z_i - \mu_z)\,\hat{\omega}_i &= \kappa_z + \sigma^2 f'(z_i)\,, \label{eq:omega_update}\\
(\beta_j - \mu_\beta)\,\hat{\lambda}_j &= \kappa_\beta + \tau^2 g'(\beta_j)\,. \label{eq:lambda_update}
\end{align}
\end{proposition}

The weight $\hat{\omega}_i$ assigned to observation $i$ is the ratio of the gradient $f'(z_i)$ to the residual $z_i - \mu_z$.  The algorithm never requires the full mixing distribution $P(\omega_i)$---only these first moments.

\begin{proposition}[Weighted least squares M-step; Proposition~2 of \citet{polson2013data}]\label{prop:mstep}
Given $\hat{\Omega} = \mathrm{diag}(\hat{\omega}_1, \ldots, \hat{\omega}_n)$ and $\hat{\Lambda} = \mathrm{diag}(\hat{\lambda}_1, \ldots, \hat{\lambda}_p)$, the M-step update is
\begin{equation}\label{eq:mstep}
\hat{\beta} = \left(\tau^{-2}\hat{\Lambda} + X^\top \hat{\Omega}\, X\right)^{-1}\!\left(X^\top(\mu_z\hat{\omega} + \kappa_z\mathbf{1}) + \tau^{-2}(\mu_\beta\hat{\lambda} + \kappa_\beta\mathbf{1})\right).
\end{equation}
\end{proposition}

The E-step computes $\hat{\omega}_i$ and $\hat{\lambda}_j$ via \eqref{eq:omega_update}--\eqref{eq:lambda_update}; the M-step solves the resulting weighted least squares~\eqref{eq:mstep}.  No step size is specified by the user---the latent variables determine the effective step size at each iteration.

\section{Replacing Robbins--Monro with Scale Mixtures}\label{sec:smrm}

\subsection{The Latent Variable Step Size}

The Robbins--Monro update~\eqref{eq:rm} can be written as $\theta_{t+1} = \theta_t - \alpha_t\,\nabla \ell(\theta_t)$, where $\alpha_t$ is a scalar step size.  In the scale mixture framework, the M-step update~\eqref{eq:mstep} is a regularized Newton step with an iteration-dependent, observation-specific precision matrix $\hat{\Omega}_t$.  Each diagonal entry $\hat{\omega}_{i,t}$ adapts to the local curvature of the loss at observation $i$.

To see the analogy precisely, consider the scalar case.  The E-step gives $\hat{\omega}_t = f'(z_t)/(z_t - \mu_z)$ and the M-step gives
\begin{equation}\label{eq:scalar_update}
\beta_{t+1} = \frac{\hat{\omega}_t\,z_t + \tau^{-2}\hat{\lambda}_t\,\mu_\beta}{\hat{\omega}_t + \tau^{-2}\hat{\lambda}_t}\,,
\end{equation}
which can be rewritten as
\begin{equation}\label{eq:rm_form}
\beta_{t+1} = \beta_t - \underbrace{\frac{1}{\hat{\omega}_t + \tau^{-2}\hat{\lambda}_t}}_{\text{effective step size}}\left(\hat{\omega}_t(\beta_t - z_t) + \tau^{-2}\hat{\lambda}_t(\beta_t - \mu_\beta)\right).
\end{equation}
The effective step size $(\hat{\omega}_t + \tau^{-2}\hat{\lambda}_t)^{-1}$ is adaptive: when $\hat{\omega}_t$ is large (strong gradient signal), the step is small and precise; when $\hat{\omega}_t$ is small (weak signal), the prior dominates and regularizes.  This is a Robbins--Monro recursion where the step size is determined by the model.

\subsection{Connection to Adam}\label{sec:adam_connection}

Adam's per-parameter scaling $1/(\sqrt{\hat{v}_t} + \epsilon)$ approximates the inverse square root of the second moment of the gradient.  In the scale mixture framework, the role of $\hat{v}_t$ is played by $\hat{\omega}_t$.

\begin{remark}[Adam as a heuristic approximation]
Consider the scale mixture update with $\hat{\omega}_i = f'(z_i)/(z_i - \mu_z)$ and the Adam update with $\hat{v}_t = \beta_2\hat{v}_{t-1} + (1-\beta_2)g_t^2$.  Both produce per-parameter scaling of the gradient.  The difference is the source of the scaling:
\begin{itemize}
\item Adam estimates $\mathbb{E}[g_t^2]$ via an exponential moving average---a heuristic that treats all loss functions identically.
\item The scale mixture approach computes $\hat{\omega}_t$ from the \emph{structure} of the loss function.  For logistic loss, $\hat{\omega}_i = \tanh(z_i/2)/(2z_i)$; for quantile loss, $\hat{\omega}_i = |y_i - x_i^\top\beta|^{-1}$.  Each loss function yields its own natural curvature estimate, with no hyperparameters ($\beta_1, \beta_2, \epsilon$) to tune.
\end{itemize}
\end{remark}

\subsection{Connection to AdamW}\label{sec:adamw_connection}

AdamW uses a fixed weight decay $\lambda_{\mathrm{wd}}$ applied to all parameters.  In the scale mixture framework, weight decay arises from the penalty term $g(\beta_j)$ via the latent variable $\hat{\lambda}_j$.

From Proposition~\ref{prop:gradient}, the adaptive weight decay for parameter $j$ is
\begin{equation}\label{eq:adaptive_wd}
\hat{\lambda}_j = \frac{\kappa_\beta + \tau^2 g'(\beta_j)}{\beta_j - \mu_\beta}\,.
\end{equation}
For an $L^2$ penalty $g(\beta) = \beta^2/2$, this gives $\hat{\lambda}_j = \tau^2$, recovering fixed weight decay.  For non-quadratic penalties, $\hat{\lambda}_j$ is parameter-specific and iteration-dependent:
\begin{itemize}
\item \emph{Lasso} ($g(\beta) = |\beta|$): $\hat{\lambda}_j = \tau^2/|\beta_j|$, which applies strong decay to small coefficients and weak decay to large ones.
\item \emph{Double-Pareto} ($g(\beta) = \gamma\log(1 + |\beta|/a)$): $\hat{\lambda}_j = \gamma\tau^2/\!\bigl((a + |\beta_j|)\,|\beta_j|\bigr)$, which interpolates between $L^1$ and $L^0$ behavior.
\item \emph{Horseshoe}: heavy-tailed $P(\lambda_j)$ allows $\hat{\lambda}_j$ to be arbitrarily small for large $|\beta_j|$, providing near-zero shrinkage for signals.
\end{itemize}
For the Lasso and double-Pareto, $\hat{\lambda}_j \to \infty$ as $\beta_j \to 0$, signaling that the coefficient should be removed from the active set.  In practice, components with $|\beta_j| < \delta$ (for a small threshold $\delta$, e.g.\ $10^{-8}$) are set to zero and excluded from the M-step linear system; see also \citet{polson2013data}, \S3.2.

In each case, the scale mixture framework replaces AdamW's single hyperparameter $\lambda_{\mathrm{wd}}$ with an adaptive, per-parameter weight decay derived from the prior.

\subsection{The Algorithm}

Combining the developments above yields the Scale Mixture EM algorithm for minimizing~\eqref{eq:objective}.

\begin{algorithm}[H]
\caption{Scale Mixture EM (SM-EM)}
\label{alg:smem}
\begin{algorithmic}[1]
\REQUIRE Data $(X, y)$, initial $\beta^{(0)}$, prior parameters $(\tau, \mu_\beta, \kappa_\beta)$, loss parameters $(\sigma, \mu_z, \kappa_z)$
\FOR{$t = 0, 1, 2, \ldots$}
\STATE \textbf{E-step (adaptive weights):}
\STATE \quad Compute residuals $z_i^{(t)} = y_i - x_i^\top \beta^{(t)}$ (or $z_i = x_i^\top\beta^{(t)}$ for classification)
\STATE \quad $\hat{\omega}_i^{(t)} = (\kappa_z + \sigma^2 f'(z_i^{(t)}))/(z_i^{(t)} - \mu_z)$ \hfill [observation weights]
\STATE \quad $\hat{\lambda}_j^{(t)} = (\kappa_\beta + \tau^2 g'(\beta_j^{(t)}))/(\beta_j^{(t)} - \mu_\beta)$ \hfill [parameter weights / adaptive decay]
\STATE \textbf{M-step (weighted least squares):}
\STATE \quad $\hat{\Omega}^{(t)} = \mathrm{diag}(\hat{\omega}_1^{(t)}, \ldots, \hat{\omega}_n^{(t)})$
\STATE \quad $\hat{\Lambda}^{(t)} = \mathrm{diag}(\hat{\lambda}_1^{(t)}, \ldots, \hat{\lambda}_p^{(t)})$
\STATE \quad $\beta^{(t+1)} = (\tau^{-2}\hat{\Lambda}^{(t)} + X^\top \hat{\Omega}^{(t)} X)^{-1}(X^\top(\mu_z\hat{\omega}^{(t)} + \kappa_z\mathbf{1}) + \tau^{-2}(\mu_\beta\hat{\lambda}^{(t)} + \kappa_\beta\mathbf{1}))$
\ENDFOR
\end{algorithmic}
\end{algorithm}

For large-scale problems, the M-step can use a mini-batch $E_t \subset \{1, \ldots, n\}$, yielding a stochastic~EM that maintains the Robbins--Monro structure.  When the M-step is partitioned into blocks---updating subsets of $\beta$ conditionally---the algorithm becomes an ECM procedure, which retains nonincreasing objective values while simplifying each conditional maximization step.

\subsection{Convergence}

The EM monotonicity property \citep{dempster1977maximum} guarantees that the objective decreases at each iteration: $L(\beta^{(t+1)}) \leq L(\beta^{(t)})$.  This holds for the full-batch algorithm without any step size condition.  For the stochastic variant, convergence follows from the online EM theory of \citet{cappe2009online}, which establishes that the stochastic EM iterates converge to stationary points of the Kullback--Leibler divergence at the rate of the maximum likelihood estimator.

The connection to classical Robbins--Monro is structural: the SM-EM iterates have the same form as a Robbins--Monro recursion, but the convergence mechanism differs.

\begin{proposition}[Robbins--Monro form; scalar case]\label{prop:rm_convergence}
In the scalar case ($p = 1$), define $\alpha_t \defeq (\hat{\omega}_t + \tau^{-2}\hat{\lambda}_t)^{-1}$.  The scale mixture iterates~\eqref{eq:rm_form} have the form
\[
\beta_{t+1} = \beta_t - \alpha_t\,H(\beta_t)\,,
\]
where $H(\beta^\star) = 0$ at the MAP estimate.  The step sizes satisfy:
\begin{enumerate}
\item $\alpha_t > 0$ for all $t$;
\item $\alpha_t \to (\omega^\star + \tau^{-2}\lambda^\star)^{-1} > 0$ as $\beta_t \to \beta^\star$.
\end{enumerate}
\end{proposition}

\begin{proof}
For part~(1), both $\hat{\omega}_t$ and $\hat{\lambda}_t$ are conditional expectations of latent variables with strictly positive support.  For the P\'olya--Gamma distribution, $\mathbb{E}[\omega \mid z] = \tanh(z/2)/(2z) > 0$ for all $z \neq 0$, so $\alpha_t = (\hat{\omega}_t + \tau^{-2}\hat{\lambda}_t)^{-1} > 0$.

For part~(2), the map $\beta \mapsto \hat{\omega}(\beta)$ is continuous by the dominated convergence theorem.  If $\beta_t \to \beta^\star$, then $\hat{\omega}_t \to \omega^\star > 0$ and $\hat{\lambda}_t \to \lambda^\star > 0$, giving $\alpha_t \to (\omega^\star + \tau^{-2}\lambda^\star)^{-1} > 0$.

For the recursion form, the M-step update~\eqref{eq:scalar_update} gives $(\hat{\omega}_t + \tau^{-2}\hat{\lambda}_t)\beta_{t+1} = \hat{\omega}_t z_t + \tau^{-2}\hat{\lambda}_t \mu_\beta$.  Subtracting $(\hat{\omega}_t + \tau^{-2}\hat{\lambda}_t)\beta_t$ and dividing yields $\beta_{t+1} - \beta_t = -\alpha_t H(\beta_t)$ with $H(\beta) = \hat{\omega}(\beta - z) + \tau^{-2}\hat{\lambda}(\beta - \mu_\beta)$.  At $\beta^\star$, the gradient of the regularized log-likelihood vanishes, so $H(\beta^\star) = 0$.
\end{proof}

\begin{remark}
Since $\alpha_t \to c > 0$, the classical Robbins--Monro conditions $\sum_t \alpha_t^2 < \infty$ do \emph{not} hold for the full-batch algorithm.  This is not a deficiency: full-batch SM-EM converges by the EM monotonicity property (each iteration decreases the objective $L(\beta)$), not by stochastic approximation theory.  The RM form is structural---it reveals that the latent variables play the role of adaptive step sizes---rather than the convergence mechanism.  For a stochastic variant using mini-batch sufficient statistics with a diminishing gain sequence $\eta_t$ satisfying $\sum_t \eta_t = \infty$, $\sum_t \eta_t^2 < \infty$, Polyak--Juditsky averaging yields the optimal $O(1/T)$ rate \citep{polyak1992acceleration, kushner2003stochastic}.
\end{remark}

\section{Proximal Operators and MAP Estimation}\label{sec:proximal}

\subsection{Convex Analysis Preliminaries}

The connection between scale mixtures and proximal algorithms rests on several facts from convex analysis (see, e.g., \citealt{boyd2004convex}, \S3.3; \citealt{rockafellar1998variational}).

The \emph{convex conjugate} of a function $f : \mathcal{R}^d \to \overline{\mathcal{R}}$ is $f^\star(\lambda) = \sup_x\{\lambda^\top x - f(x)\}$.  As $f^\star$ is the pointwise supremum of a family of affine functions, it is convex even when $f$ is not.

\begin{lemma}[Fenchel--Moreau duality; \citealt{boyd2004convex}, \S3.3.2]\label{lem:fenchel}
Let $f : \mathcal{R}^d \to \overline{\mathcal{R}}$ be a closed convex function.  Then there exists a convex function $f^\star(\lambda)$ such that
\begin{align}
f(x) &= \sup_\lambda\left\{\lambda^\top x - f^\star(\lambda)\right\}, \label{eq:fenchel_f}\\
f^\star(\lambda) &= \sup_x\left\{\lambda^\top x - f(x)\right\}. \label{eq:fenchel_fstar}
\end{align}
If $f(x)$ is instead a concave function, then $\sup$ is replaced by $\inf$ in both equations.  Any maximizing value of $\lambda$ in~\eqref{eq:fenchel_f} satisfies $\hat{\lambda} \in \partial f(x)$; if $f$ is differentiable, $\hat{\lambda} = \nabla f(x)$.
\end{lemma}

One function $g(x)$ is said to \emph{majorize} $f(x)$ at $x_0$ if $g(x_0) = f(x_0)$ and $g(x) \geq f(x)$ for all $x \neq x_0$.  A function $f(x)$ is \emph{completely monotone} on $A \subset \mathcal{R}$ if its derivatives alternate in sign: $(-1)^k f^{(k)}(x) \geq 0$ for all $k = 0, 1, 2, \ldots$ and all $x \in A$.

For any function $f(x)$, the \emph{Moreau envelope} $E_\gamma f(x)$ and \emph{proximal mapping} $\prox_{\gamma f}(x)$ for parameter $\gamma > 0$ are
\begin{align}
E_\gamma f(x) &= \inf_z\left\{f(z) + \frac{1}{2\gamma}\|z - x\|_2^2\right\} \leq f(x)\,, \label{eq:moreau_env}\\
\prox_{\gamma f}(x) &= \arg\min_z\left\{f(z) + \frac{1}{2\gamma}\|z - x\|_2^2\right\}. \label{eq:prox_def}
\end{align}
The Moreau envelope is a regularized version of $f$: it approximates $f$ from below and has the same set of minimizing values as $f$ \citep{rockafellar1998variational}.  The proximal mapping returns the value that solves the minimization problem defined by the Moreau envelope---it balances two goals: minimizing $f$ and staying near $x$.  When $f(x) = \iota_C(x)$ is the indicator of a convex set $C$, $\prox_f(x) = \arg\min_{z \in C}\|x - z\|_2^2$ reduces to Euclidean projection.

\subsection{Envelopes and Half-Quadratic Representations}

The scale mixture representation has a dual interpretation via envelopes \citep{polson2016mixtures}.  A penalty function $\phi(\beta)$ with completely monotone derivative of $\phi(\sqrt{2x})$ admits the half-quadratic representation
\begin{equation}\label{eq:envelope_scale}
\phi(\beta) = \inf_{\omega \geq 0}\left\{\frac{\omega}{2}\beta^2 - \psi^\star(\omega)\right\},
\end{equation}
where $\psi^\star$ is the convex conjugate and the infimum is attained at $\hat{\omega} = \phi'(\beta)/\beta$ \citep{geman1992constrained}.  Alternatively, the Geman--Yang location representation \citep{geman1995nonlinear} gives
\begin{equation}\label{eq:envelope_location}
\phi(\beta) = \inf_{\lambda}\left\{\frac{1}{2}(\beta - \lambda)^2 + \psi(\lambda)\right\},
\end{equation}
where $\hat{\lambda} = \prox_\psi(\beta)$ is the proximal operator of $\psi$ \citep{moreau1965proximite}.

The marginalizing prior (mixture) and profiling prior (envelope) yield the same function $\phi$ but through different operations: integration versus optimization.  \citet{polson2016mixtures} call these ``hierarchical duals.''  For algorithmic purposes, the envelope representation leads to MM algorithms where each iteration minimizes a quadratic majorizer, and the mixture representation leads to EM algorithms where each iteration computes a conditional expectation.  Both yield the same iterates.

\subsection{Proximal Gradient as Scale Mixture Iteration}

For the composite objective $F(\beta) = f(\beta) + \phi(\beta)$ where $\nabla f$ is $L$-Lipschitz, the proximal gradient method iterates
\begin{equation}\label{eq:prox_grad}
\beta_{t+1} = \prox_{\gamma\phi}\!\left(\beta_t - \gamma\nabla f(\beta_t)\right), \quad \gamma = 1/L\,.
\end{equation}
This is equivalent to the scale mixture iteration with a Gaussian location envelope of $f$ \citep{polson2016mixtures}: the latent variable update $\hat{\lambda}_{t+1} = \gamma^{-1}\beta_t - \nabla f(\beta_t)$ encodes the gradient step, and the $\beta$-update is the proximal (shrinkage) step.

The MAP estimate satisfies the fixed-point equation
\begin{equation}\label{eq:fixed_point}
\beta^\star = \prox_{\gamma\phi}\!\left(\beta^\star - \gamma\nabla f(\beta^\star)\right) = (I + \gamma\,\partial\phi)^{-1}(I - \gamma\nabla f)\,\beta^\star\,,
\end{equation}
which decomposes into a forward operator $(I - \gamma\nabla f)$ (gradient step) and a backward operator $(I + \gamma\partial\phi)^{-1}$ (proximal/prior step).

Nesterov acceleration \citep{nesterov1983method} replaces $\beta_t$ in \eqref{eq:prox_grad} with an extrapolated point $v_t = \beta_t + \theta_t(\theta_{t-1}^{-1} - 1)(\beta_t - \beta_{t-1})$, achieving $O(1/t^2)$ convergence for convex objectives \citep{beck2009fast}.  This introduces non-monotone behavior: the objective can increase at individual iterates.  The base (non-accelerated) scale mixture EM avoids this by construction---the EM property guarantees monotonic descent.  When Nesterov extrapolation is applied to SM-EM, the monotonicity guarantee is lost (as for any accelerated method), but the adaptive curvature estimates from the E-step still determine per-iteration scaling without user-specified step sizes.

\subsection{Improved Step Sizes via Scale Mixtures}

Standard proximal gradient uses the global Lipschitz constant $L$ to set $\gamma = 1/L$, which is conservative when the curvature varies.  The scale mixture approach replaces this with observation-specific curvature estimates.

For the logistic loss $f(z) = \log(1 + e^z)$, the half-quadratic representation gives \citep{polson2013bayesian}
\begin{equation}\label{eq:logistic_hq}
\log\cosh(z/2) = \inf_{\omega \geq 0}\left\{\frac{\omega}{2}z^2 - \psi^\star(\omega)\right\}, \quad \hat{\omega}(z) = \frac{\tanh(z/2)}{2z}\,.
\end{equation}
The curvature estimate $\hat{\omega}(z) = \tanh(z/2)/(2z)$ is bounded between $0$ and $1/4$ and varies smoothly with $z$.  Compared to the global bound $L = 1/4$ (the maximum of $w(z)(1-w(z))$ where $w(z) = 1/(1+e^{-z})$), the local estimate $\hat{\omega}(z)$ is tighter for observations far from the decision boundary, leading to larger effective step sizes where the loss is flatter.  This improves on the standard Lipschitz bound and provides the scale mixture analog of the diagonal scaling in Adam.

\section{Empirical Comparison}\label{sec:empirical}

We compare scale mixture EM (SM-EM) with Adam, SGD with momentum, and the proximal Robbins--Monro (PRM) method of \citet{toulis2021proximal}, with and without Nesterov acceleration.  The experiments evaluate two questions: (i) whether SM-EM reaches lower loss than Adam without optimizer-schedule tuning, and (ii) whether the sufficient-statistics structure of the M-step reduces runtime for regularization paths.  All experiments use standardized synthetic logistic-regression data.  SM-EM uses P\'olya--Gamma augmentation with prior precision $\tau^{-2} = 0.01$; Adam uses $(\beta_1, \beta_2, \epsilon) = (0.9, 0.999, 10^{-8})$, batch size 256; PRM uses implicit updates with Hessian, batch size 20.  SM-EM has no learning rate; for Adam and SGD we report both default and grid-tuned learning rates.  All experiments use fixed random seeds for reproducibility.  Since SM-EM is a full-batch, deterministic algorithm, its convergence trajectory is fully determined by the data realization; the stochastic methods (Adam, SGD, PRM) are also deterministic for a given seed due to fixed mini-batch orderings.  Results are from single data realizations; the comparison targets algorithmic convergence behavior rather than statistical variability across datasets.  We omit IRLS (Fisher scoring) as a baseline because it solves the same $p \times p$ linear system per iteration as SM-EM; the comparison of interest is whether SM-EM's model-derived weights outperform heuristic adaptive methods (Adam, SGD) that avoid the full Hessian.

\subsection{Convergence and Learning Rate Sensitivity}

Figure~\ref{fig:convergence} shows the negative log-likelihood versus iteration for a moderately conditioned problem ($\mathrm{cond}(X^\top X) \approx 50$).  SM-EM reaches the lowest loss in this setting ($0.089$) without learning-rate tuning.  Adam with the default $\alpha = 10^{-3}$ has not converged after 80 epochs; even with tuned $\alpha = 10^{-2}$, Adam converges more slowly (final loss $0.110$).  PRM converges but requires a tuned step-size schedule.  Figure~\ref{fig:sensitivity} sweeps Adam's learning rate over $[10^{-4}, 10^{-0.5}]$: performance varies by about a factor of six, while SM-EM matches the best result in this sweep without optimizer tuning.

\begin{figure}[t]
\centering
\includegraphics[width=0.75\textwidth]{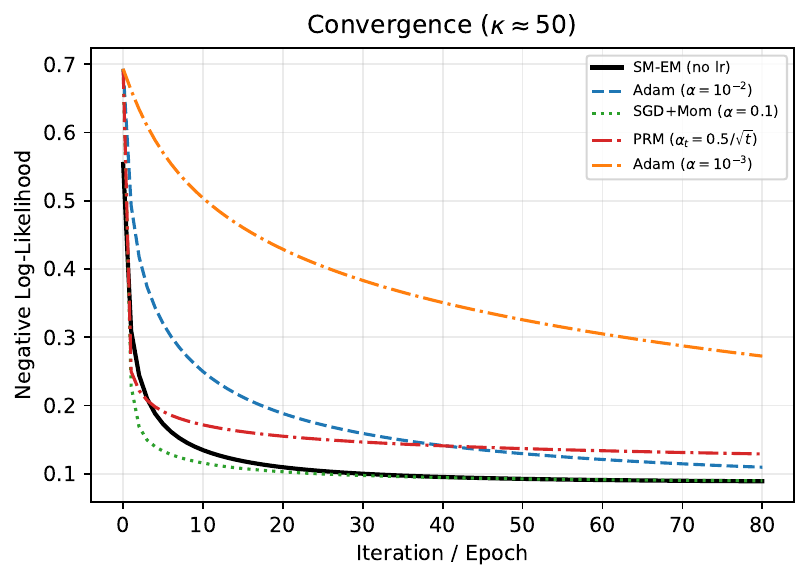}
\caption{Convergence on moderately conditioned logistic regression ($\mathrm{cond} \approx 50$).  SM-EM (no learning rate) achieves the lowest loss.  PRM (implicit SGD) converges but requires step size tuning.}
\label{fig:convergence}
\end{figure}

\begin{figure}[t]
\centering
\includegraphics[width=0.75\textwidth]{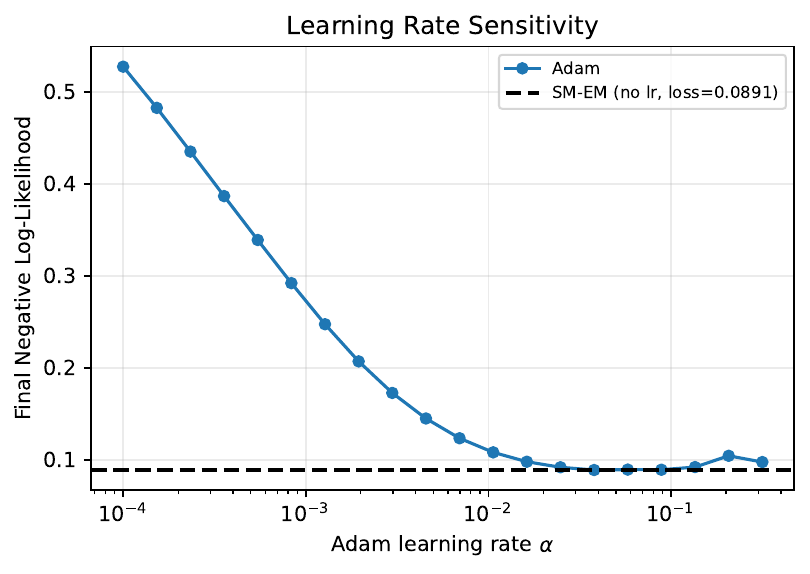}
\caption{Adam's final loss as a function of learning rate $\alpha$.  SM-EM (dashed line) matches the best Adam result in this sweep without learning-rate tuning.}
\label{fig:sensitivity}
\end{figure}

\subsection{Effect of Conditioning}

Figure~\ref{fig:ill_conditioned} compares SM-EM, PRM, and Adam on a well-conditioned ($\mathrm{cond} \approx 50$) and an ill-conditioned ($\mathrm{cond} \approx 500$) problem.  SM-EM performs best in these experiments and is less sensitive to conditioning, because the P\'olya--Gamma weights $\hat{\omega}_i$ provide local curvature estimates that adapt by observation.  Adam's fixed learning rate becomes increasingly suboptimal as the condition number grows.  PRM, which uses Hessian information, handles ill-conditioning better than Adam but still requires a step-size schedule.

\begin{figure}[t]
\centering
\includegraphics[width=\textwidth]{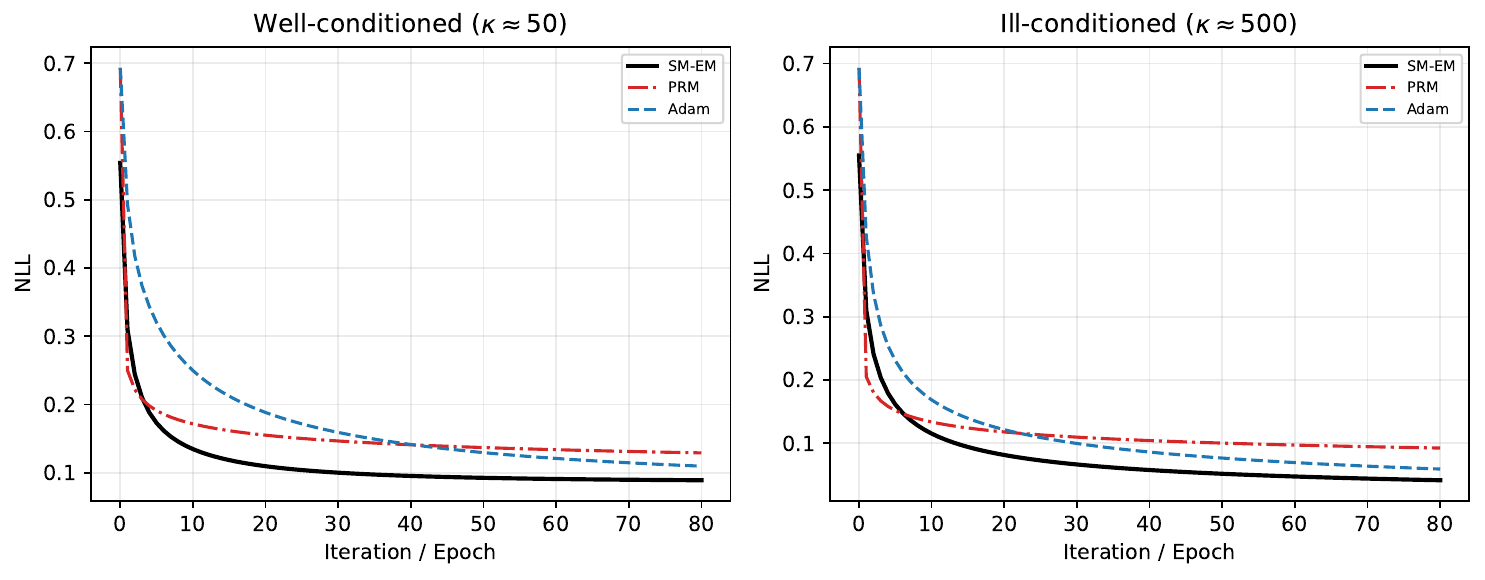}
\caption{Effect of conditioning.  Left: $\mathrm{cond} \approx 50$.  Right: $\mathrm{cond} \approx 500$.  SM-EM converges robustly in both cases; Adam and PRM degrade as conditioning worsens.}
\label{fig:ill_conditioned}
\end{figure}

\subsection{Nesterov Acceleration}

Since each full-batch SM-EM iteration is deterministic, Nesterov extrapolation can be applied without variance reduction.  Figure~\ref{fig:nesterov} shows the effect.  On the well-conditioned problem ($\mathrm{cond} \approx 50$), Nesterov provides modest improvement for both SM-EM and PRM.  On the ill-conditioned problem ($\mathrm{cond} \approx 500$), the effect is substantial: SM-EM+Nesterov achieves a 55\% lower final loss than SM-EM alone (NLL $0.019$ vs.\ $0.042$), and PRM+Nesterov improves by 54\% over standard PRM.  The observed speedup is consistent with the $O(\sqrt{L/\mu})$ dependence predicted by Nesterov theory for fixed-step methods (Appendix~\ref{app:convergence}), though the formal rate guarantee for adaptive curvature remains open.

A separate experiment isolating SM-EM and PRM with Nesterov on the same ill-conditioned design ($\mathrm{cond} \approx 450$, $n=2000$, $p=20$) gives similar conclusions (Table~\ref{tab:nesterov}).  SM-EM+Nesterov reduces final NLL by 32\% relative to standard SM-EM ($0.030$ vs.\ $0.044$), while PRM+Nesterov achieves a 53\% reduction ($0.063$ vs.\ $0.135$).  A plausible mechanism is that SM-EM's closed-form M-step already exploits current curvature information, so Nesterov extrapolation mainly improves trajectory efficiency.  PRM also benefits, although mini-batch noise can reduce look-ahead accuracy.

\begin{table}[t]
\centering
\caption{Nesterov acceleration on an ill-conditioned problem ($\mathrm{cond} \approx 450$, $n{=}2000$, $p{=}20$, 100 iterations).}\label{tab:nesterov}
\begin{tabular}{lccc}
\hline
\textbf{Method} & NLL & Accuracy & Nesterov gain \\
\hline
SM-EM             & 0.044 & 98.0\% & ---   \\
SM-EM + Nesterov  & \textbf{0.030} & 97.8\% & 32\%  \\
PRM               & 0.135 & 96.8\% & ---   \\
PRM + Nesterov    & 0.063 & 98.0\% & 53\%  \\
\hline
\end{tabular}
\end{table}

\begin{figure}[t]
\centering
\includegraphics[width=\textwidth]{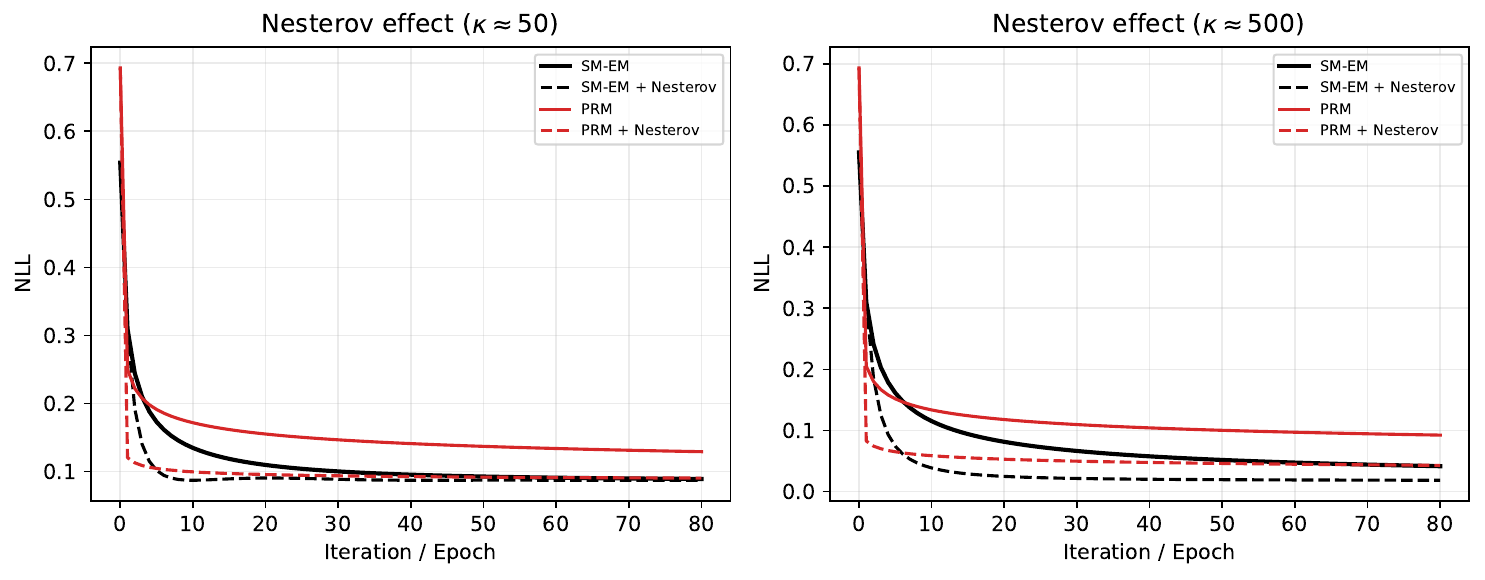}
\caption{Nesterov acceleration.  Left: $\mathrm{cond} \approx 50$ (modest effect).  Right: $\mathrm{cond} \approx 500$ (55\% for SM-EM, 54\% for PRM).  Table~\ref{tab:nesterov} reports a separate ill-conditioned experiment with 32\% (SM-EM) and 53\% (PRM) gains.}
\label{fig:nesterov}
\end{figure}

\subsection{Scaling with Dimension}

Table~\ref{tab:highdim} and Figure~\ref{fig:high_dim} report a sweep over $p \in \{20, 50, 100, 200, 500\}$ with $n = 5{,}000$, $\mathrm{cond} = 500$, and 80 iterations.  For Adam we report both the default learning rate $\alpha = 10^{-3}$ (no tuning) and the best learning rate selected by grid search over six values; for SGD+Momentum we report the grid-search best.  SM-EM and SM-EM+Nesterov use no learning rate at all.

Two patterns emerge.  First, SM-EM without learning-rate tuning outperforms Adam with default tuning by $4$--$6\times$ in NLL across all dimensions: Adam at $\alpha = 10^{-3}$ converges slowly on these ill-conditioned problems, while SM-EM reaches lower losses automatically.  Second, SM-EM+Nesterov outperforms grid-tuned Adam at every value of $p$---by $4\times$ at $p = 200$ (NLL $0.007$ vs.\ $0.028$) and $13\times$ at $p = 500$ (NLL $0.0015$ vs.\ $0.019$).  The gap grows with $p$ because SM-EM's M-step uses the full curvature matrix $X^\top \hat{\Omega} X$, whereas Adam's diagonal second-moment estimate $\hat{v}_t$ omits off-diagonal Hessian structure.

The cost per iteration is higher: the $O(p^3)$ linear solve makes SM-EM slower in wall-clock time for $p \geq 200$ (Table~\ref{tab:highdim}).  However, the total practitioner cost is lower: SM-EM runs once, while Adam requires a grid search over learning rates---each point on the grid takes as long as a single SM-EM run.  The Halton Monte Carlo method described below further reduces SM-EM's per-iteration cost for large $p$.

\begin{table}[t]
\centering
\caption{Final NLL and wall-clock time vs.\ dimension ($n{=}5000$, $\mathrm{cond}{=}500$, 80 iterations).  Adam$^*$ uses grid-tuned learning rate.}\label{tab:highdim}
\begin{tabular}{rcccccccc}
\hline
 & \multicolumn{2}{c}{SM-EM} & \multicolumn{2}{c}{SM-EM+Nest} & \multicolumn{2}{c}{Adam$^*$} & \multicolumn{2}{c}{Adam ($10^{-3}$)} \\
\cmidrule(lr){2-3}\cmidrule(lr){4-5}\cmidrule(lr){6-7}\cmidrule(lr){8-9}
$p$ & NLL & Time & NLL & Time & NLL & Time & NLL & Time \\
\hline
20  & 0.050 & 0.03 & \textbf{0.036} & 0.03 & 0.040 & 0.04 & 0.233 & 0.04 \\
50  & 0.043 & 0.04 & \textbf{0.017} & 0.04 & 0.026 & 0.06 & 0.167 & 0.06 \\
100 & 0.042 & 0.10 & \textbf{0.019} & 0.10 & 0.029 & 0.12 & 0.157 & 0.12 \\
200 & 0.036 & 0.34 & \textbf{0.007} & 0.34 & 0.028 & 0.10 & 0.137 & 0.10 \\
500 & 0.021 & 1.29 & \textbf{0.0015} & 1.28 & 0.019 & 0.30 & 0.116 & 0.25 \\
\hline
\end{tabular}
\end{table}

\begin{figure}[t]
\centering
\includegraphics[width=\textwidth]{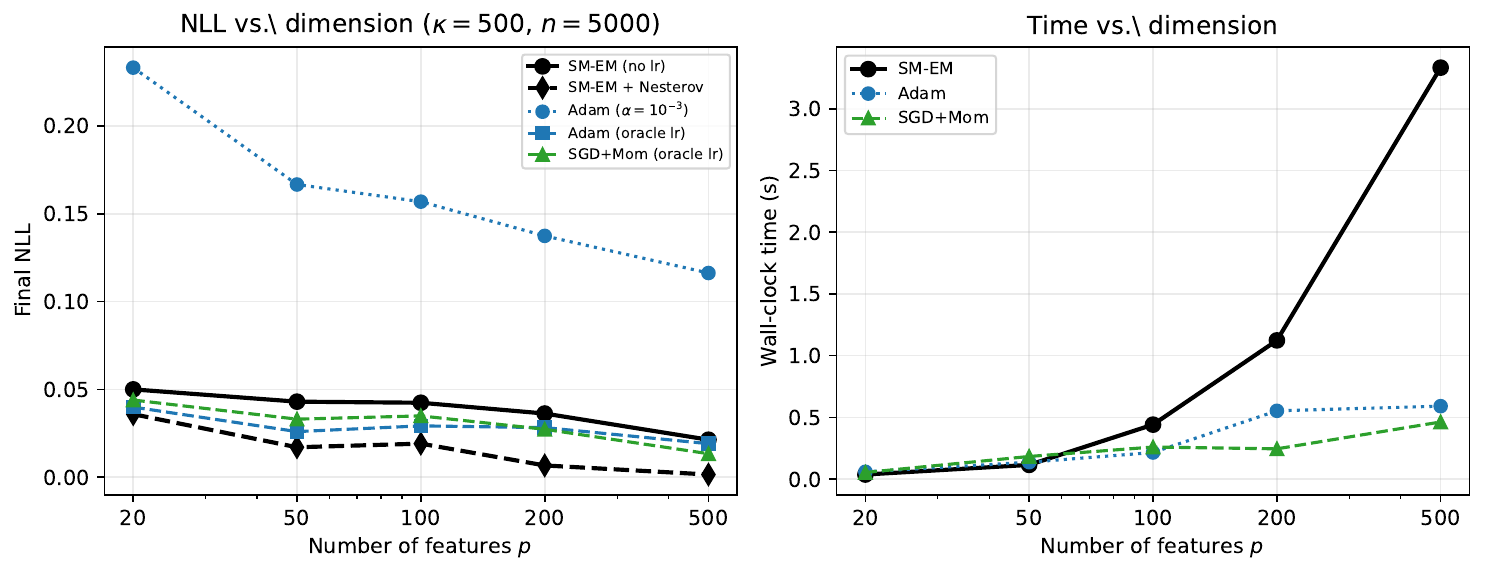}
\caption{Left: Final NLL vs.\ $p$ ($\mathrm{cond}{=}500$, $n{=}5000$).  SM-EM+Nesterov (no learning-rate tuning) reaches lower NLL than grid-tuned Adam at every dimension tested; the gap widens with $p$.  Adam at its default learning rate ($\alpha{=}10^{-3}$) has not converged in 80 epochs.  Right: wall-clock time.  SM-EM's $O(p^3)$ solve is the dominant cost for $p \geq 200$.}
\label{fig:high_dim}
\end{figure}

\subsection{Computational Cost and Scalability}

Table~\ref{tab:results} summarizes final losses and wall-clock times.  SM-EM is roughly $20\times$ faster than PRM per iteration: the M-step solves a $p \times p$ linear system ($O(p^2 n + p^3)$), while PRM must solve $n/b$ implicit Newton steps per epoch, each involving a Hessian computation ($O(p^2 b + p^3)$ per batch).  Adam and SGD are faster per epoch due to simpler updates, but require careful learning rate tuning.  With grid-tuned step sizes, SGD+Momentum matches SM-EM's loss at $\mathrm{cond} \approx 500$ (Table~\ref{tab:results}); SM-EM's advantage is that it achieves comparable loss automatically, with no step size to select.

\paragraph{Halton Monte Carlo for the M-step.}
For large $p$, the $O(p^3)$ direct solve in the M-step~\eqref{eq:mstep} can become a bottleneck.  \citet{halton1994sequential} introduced a sequential Monte Carlo method for linear systems that exploits the Neumann series representation: if the M-step system $A\beta = b$ is written in fixed-point form $\beta = c + H\beta$ with spectral radius $\rho_H < 1$, Halton's method iteratively estimates and corrects the residual, achieving geometric convergence at rate $\rho_H^N$ \citep{halton2008practical}.  A subsampled variant reduces the per-iteration cost from $O(p^2 s)$ to $O(ps)$ for $s$ correction steps \citep{polson2026fast}.  In a scaling study on diagonally dominant systems at $p = 5{,}000$, the sampled sequential method is approximately $24\times$ faster than Gauss--Seidel \citep{polson2026fast}; this factor is implementation-dependent and illustrates asymptotic scaling rather than a portable constant.

\paragraph{Complete sufficient statistics as offline preconditioning.}
A further computational advantage arises from the structure of the M-step, and it is best understood as a form of preconditioning.  As noted by \citet{polson2011data} in the context of SVMs, the complete-data sufficient statistics $X^\top \hat{\Omega} X$ and $X^\top \kappa$ decompose as sums over observations: $X^\top \hat{\Omega} X = \sum_{i=1}^n \hat{\omega}_i\,x_i x_i^\top$ and $X^\top \kappa = \sum_{i=1}^n \kappa_i\,x_i$.  The rank-one outer products $G_i = x_i x_i^\top$ do not depend on $\beta$ or the iteration index $t$; they can be precomputed once offline at cost $O(np^2)$ and cached.  At each E-step, updating $\hat{\Omega}$ requires only $O(n)$ scalar operations (evaluating $\hat{\omega}_i$ for each observation), after which the M-step matrix $A = \tau^{-2}\hat{\Lambda} + \sum_i \hat{\omega}_i G_i$ is assembled by a weighted sum of the cached outer products at cost $O(np^2)$---the same as a single matrix--vector product, with no redundant feature-matrix operations.

The analogy to preconditioning is close.  In iterative linear algebra, a preconditioner $P$ is an approximation to the system matrix that is computed once (offline) and applied at each iteration to accelerate convergence: instead of solving $Ax = b$, one solves the better-conditioned system $P^{-1}Ax = P^{-1}b$.  The cached outer products $\{G_i\}_{i=1}^n$ play an analogous role for SM-EM.  The ``preconditioner'' here is the full set of data-dependent building blocks from which every iteration's system matrix is assembled; only the scalar weights $\hat{\omega}_i^{(t)}$ change across iterations.  Standard preconditioners (incomplete Cholesky, Jacobi, etc.) approximate curvature once and hold it fixed.  In SM-EM, recomputing $\hat{\omega}_i$ at each E-step makes the effective preconditioner $\hat{\Omega}^{(t)}$ iteration dependent.  This is the computational counterpart of the statistical observation that latent variables provide local curvature estimates (Section~\ref{sec:adam_connection}).

The offline/online decomposition also explains why SM-EM is cheap to warm-start.  When the model changes---a new observation arrives, or the regularization parameter $\tau$ is updated---only the $O(n)$ weights $\hat{\omega}_i$ and the $O(p)$ weights $\hat{\lambda}_j$ need to be refreshed; the cached $\{G_i\}$ remain valid.  By contrast, gradient-based methods must discard their momentum buffers ($m_t$, $v_t$ in Adam) and effectively restart, because those buffers encode a running average of stale gradients.

For regularized models with adaptive penalties, a dimension-reduction effect provides further savings.  As coefficients $\beta_j \to 0$ under Lasso-type penalties, the adaptive weight $\hat{\lambda}_j = \tau^2/|\beta_j| \to \infty$, effectively removing coordinate $j$ from the active set.  The M-step matrix then operates on the reduced set of active variables $\gamma = \{j : |\beta_j| > \delta\}$, and the sufficient statistics restricted to $\gamma$ are submatrices of the precomputed quantities: $(X^\top \hat{\Omega} X)_\gamma = X_\gamma^\top \hat{\Omega} X_\gamma$ \citep{polson2011data}.  Early iterations (large active set) are most expensive; later iterations (small active set) are cheap.  This makes model exploration---evaluating many candidate penalty levels $\tau$---very fast, since the sufficient statistics are shared across all penalty values.

\paragraph{Regularization path via amortized E-step.}
The shared-sufficient-statistics structure is especially powerful when solving a regularization path---evaluating penalty values $\tau_1^{-2} > \tau_2^{-2} > \cdots > \tau_K^{-2}$ to select a prior precision.  At each EM iteration, the E-step computes $\hat{\Omega}$ once at cost $O(np^2)$; the resulting matrix $X^\top \hat{\Omega} X$ is then shared across all $K$ penalty values, each requiring only an $O(p^3)$ linear solve.  The total cost per EM sweep is $O(np^2 + K p^3)$, compared with $O(K\,np^2)$ if each penalty required a separate E-step, or $O(K \cdot T \cdot np / b)$ for Adam running $T$ epochs with batch size $b$ on each of $K$ trajectories.

Figure~\ref{fig:regpath} demonstrates this on a 40-point regularization grid ($n = 5{,}000$, $p = 200$, $\mathrm{cond} = 200$, 30 EM iterations).  The amortized SM-EM path completes in $0.5$~s---$8\times$ faster than running individual SM-EM fits, $3.5\times$ faster than Adam, and $10\times$ faster than grid-tuned Adam (which must try four learning rates at each $\tau^{-2}$).  The amortized path also achieves a $5.5\times$ lower mean NLL than Adam ($0.057$ vs.\ $0.315$), because SM-EM benefits from the full Hessian at every $\tau^{-2}$ while Adam's diagonal approximation degrades on the ill-conditioned design.

Figure~\ref{fig:activeset} shows the active-set shrinkage effect with adaptive Lasso penalties ($p = 500$, $\tau = 1$).  The active set decreases from $500$ to $282$ variables over $80$ EM iterations, reducing the per-iteration solve time by $1.5\times$ as the effective dimension falls.

\begin{figure}[t]
\centering
\includegraphics[width=\textwidth]{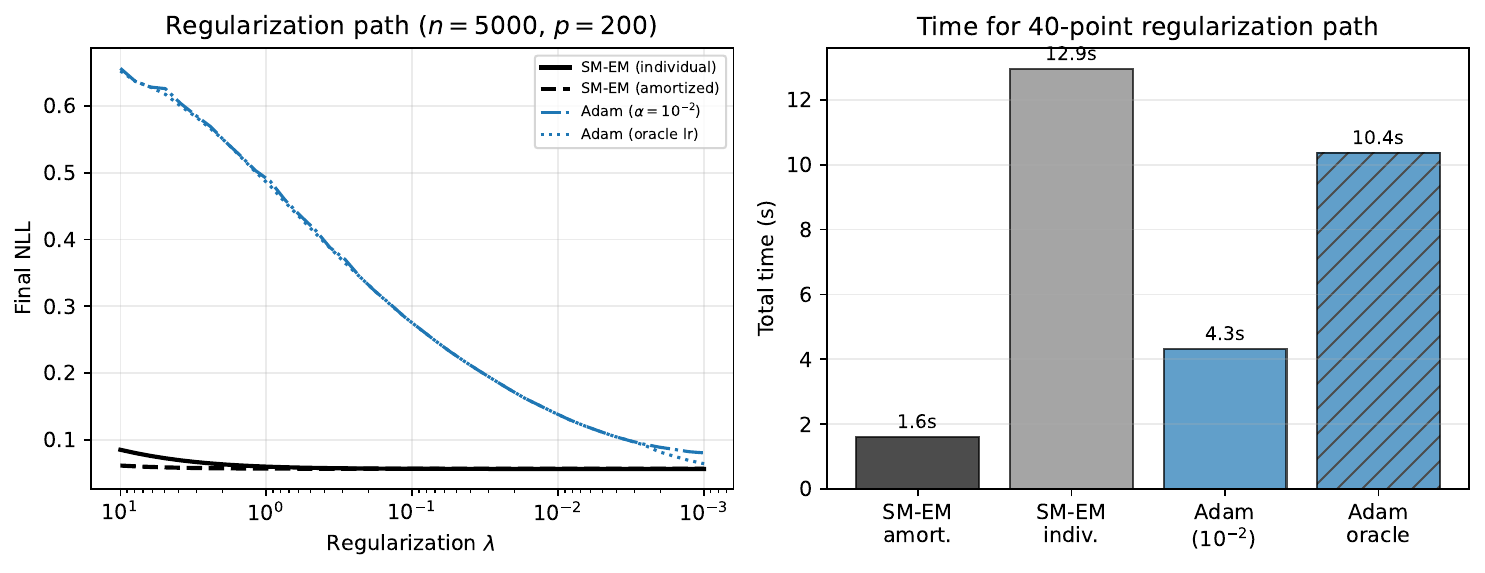}
\caption{Regularization path ($n{=}5000$, $p{=}200$, 40 penalty values).  Left: SM-EM with amortized E-step reaches the same NLL as individual SM-EM fits; Adam has not converged at most penalty values.  Right: wall-clock time comparison; amortized SM-EM completes the path in $0.5$~s.}
\label{fig:regpath}
\end{figure}

\begin{figure}[t]
\centering
\includegraphics[width=\textwidth]{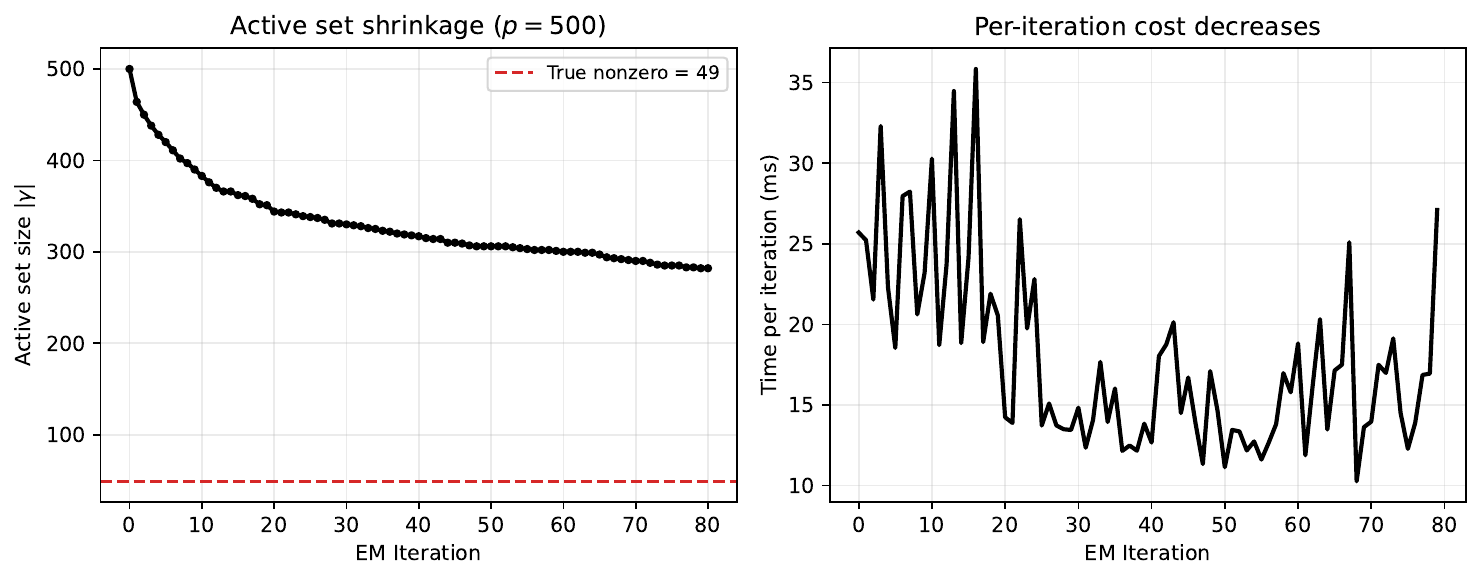}
\caption{Active-set shrinkage with adaptive Lasso ($p{=}500$, $\tau{=}1$).  Left: the active set $|\gamma|$ decreases as coefficients are driven to zero.  Right: per-iteration wall-clock time decreases proportionally, since the M-step operates on the reduced $|\gamma| \times |\gamma|$ system.}
\label{fig:activeset}
\end{figure}

\begin{table}[t]
\centering
\caption{Final negative log-likelihood and wall-clock time (80 iterations/epochs).  $^*$Grid-tuned learning rate.}\label{tab:results}
\begin{tabular}{lcccc}
\hline
 & \multicolumn{2}{c}{$\mathrm{cond} \approx 50$} & \multicolumn{2}{c}{$\mathrm{cond} \approx 500$} \\
\cmidrule(lr){2-3}\cmidrule(lr){4-5}
\textbf{Method} & NLL & Time (s) & NLL & Time (s) \\
\hline
SM-EM (no lr)         & 0.089 & 0.02 & 0.042 & 0.02 \\
SM-EM + Nesterov      & 0.087 & 0.02 & \textbf{0.019} & 0.02 \\
PRM ($\alpha_t = 0.5/\sqrt{1+t}$) & 0.129 & 0.42 & 0.093 & 0.42 \\
PRM + Nesterov        & 0.090 & 0.45 & 0.043 & 0.46 \\
Adam ($\alpha = 10^{-2}$) & 0.110 & 0.04 & 0.059 & 0.04 \\
SGD + Momentum$^*$    & 0.090 & 0.04 & 0.041 & 0.04 \\
\hline
\end{tabular}
\end{table}

\subsection{Adaptive Weights}

Figure~\ref{fig:weights} displays the SM-EM weights $\hat{\omega}_i = \tanh(z_i/2)/(2z_i)$ for five randomly chosen observations across EM iterations.  All weights start near the upper bound $1/4$ and decrease as the model classifies observations with confidence: well-separated observations ($|z_i|$ large) receive small weights, while ambiguous observations near the decision boundary retain larger weights.  The right panel shows Adam's second-moment estimate $v_t^{(j)}$: these rise then fall, driven by gradient noise rather than by the structure of the loss.

\begin{figure}[t]
\centering
\includegraphics[width=\textwidth]{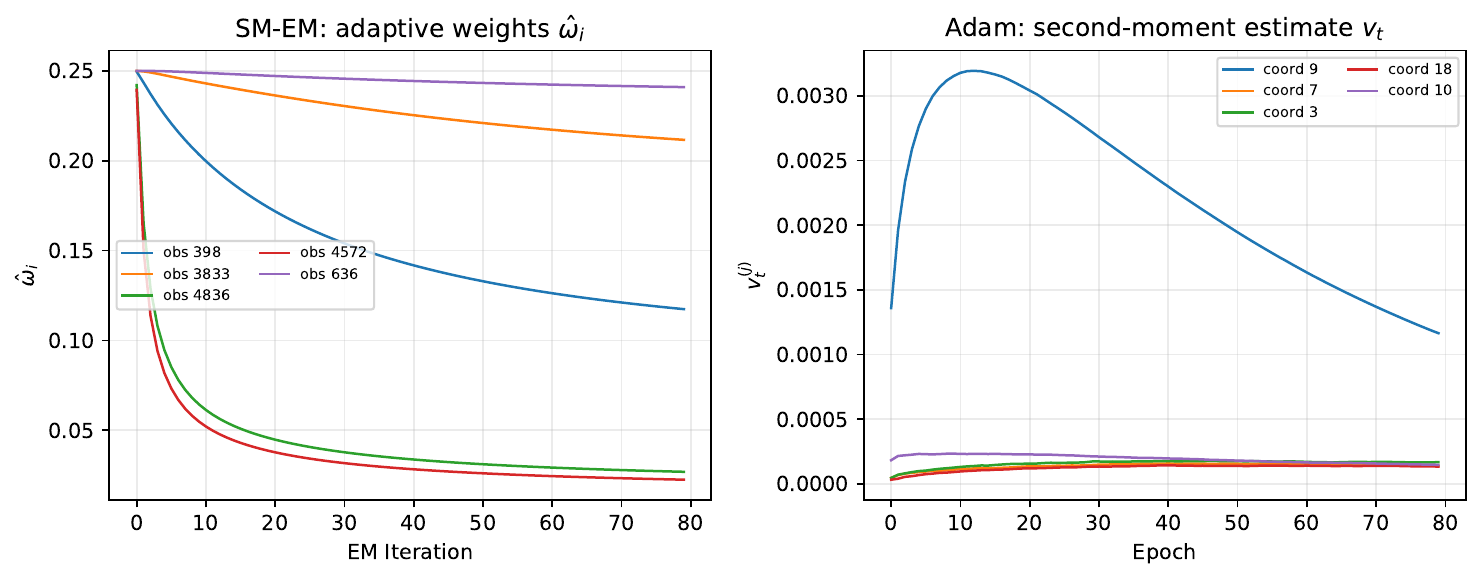}
\caption{Left: SM-EM weights $\hat{\omega}_i$ for five observations---model-derived, decreasing as observations become well-classified.  Right: Adam's $v_t^{(j)}$---heuristic, driven by gradient noise.}
\label{fig:weights}
\end{figure}

\section{Applications}\label{sec:applications}

\subsection{Logistic Regression}

For binomial logistic regression with $n$ observations, $p$ predictors, and trial counts $m_i$ ($m_i = 1$ for binary data), the negative log-likelihood is $\sum_{i=1}^n [m_i \log(1 + e^{x_i^\top\beta}) - y_i\,x_i^\top\beta]$.  Using the P\'olya--Gamma representation \citep{polson2013bayesian}, the latent variable $\omega_i$ has conditional expectation
\[
\hat{\omega}_i = \frac{m_i}{2z_i}\tanh\!\left(\frac{z_i}{2}\right), \quad z_i = x_i^\top\beta\,,
\]
and the M-step reduces to the weighted least squares
\[
\hat{\beta} = \left(\tau^{-2}\hat{\Lambda} + X^\top\hat{\Omega}\,X\right)^{-1}X^\top\kappa\,,
\]
where $\kappa_i = y_i - m_i/2$.  This is more numerically stable than standard IRLS, whose weights $w_i(1-w_i) = O(e^{-|z_i|})$ decay exponentially for well-separated observations, causing near-singular Hessians.  The P\'olya--Gamma weights $\hat{\omega}_i = O(1/|z_i|)$ decay polynomially, maintaining better-conditioned systems.

\subsection{Regularized Regression}

For the Lasso ($g(\beta) = |\beta|$), the penalty admits the scale mixture representation $|\beta| = \inf_{\lambda \geq 0}\{\beta^2/(2\lambda) + \lambda/2\}$.  Combined with a quadratic loss, the E-step via~\eqref{eq:adaptive_wd} gives $\hat{\lambda}_j = \tau^2/|\beta_j^{(t)}|$ and the M-step solves a ridge regression with precision $\tau^{-2}\hat{\Lambda}$.  This is iteratively reweighted $L^2$ penalization.  \citet{strawderman2012hierarchical} showed that hierarchical scale-mixture-of-normals priors yield a generalized grouped Lasso whose MAP estimate is a thresholding estimator with exact risk formulas, connecting the same augmentation structure to penalized likelihood theory.

For the double-Pareto penalty $g(\beta) = \gamma\log(1 + |\beta|/a)$, the proximal operator has the closed form
\[
\prox_{\gamma g}(u) = \frac{\mathrm{sgn}(u)}{2}\left\{|u| - a + \sqrt{(a - |u|)^2 + 4(a|u| - \gamma)_+}\right\},
\]
and applying~\eqref{eq:adaptive_wd} gives the adaptive weight $\hat{\lambda}_j = \gamma\tau^2/\!\bigl((a + |\beta_j|)\,|\beta_j|\bigr)$, consistent with Section~\ref{sec:adam_connection}.  This provides a smooth interpolation between $L^1$ behavior (for small $|\beta_j|$) and negligible shrinkage (for large $|\beta_j|$), replacing the fixed weight decay in AdamW with a data-adaptive alternative.

\subsection{Quantile Regression}

For quantile regression with check loss $\rho_q(\theta) = \frac{1}{2}|\theta| + (q - \frac{1}{2})\theta$, the variance-mean mixture representation with $\kappa_z = 1 - 2q$ and generalized inverse Gaussian mixing gives $\hat{\omega}_i = |y_i - x_i^\top\beta|^{-1}$.  The algorithm reduces to iteratively reweighted least squares with adaptive weights that downweight observations far from the quantile---analogous to Adam's scaling but derived from the loss geometry.

\subsection{Generalized Ridge Regression and Shrinkage Profiles}

The M-step~\eqref{eq:mstep} has a natural interpretation in the canonical basis of the design matrix.  Let $X = UDW^\top$ be the singular value decomposition and $\alpha = W^\top\beta$ the rotated coefficients.  In this basis, the least-squares estimates are $\hat{\alpha}_i = \alpha_i + \epsilon_i$ with $\epsilon_i \sim N(0, \sigma^2/d_i^2)$, where $d_i$ is the $i$th singular value.  Placing independent priors $\alpha_i \sim N(0, v_i)$ yields the generalized ridge estimator
\begin{equation}\label{eq:grr}
\alpha_i^\star = s_i\,\hat{\alpha}_i\,, \quad s_i = \frac{d_i^2\,v_i}{\sigma^2 + d_i^2\,v_i} = \frac{d_i^2}{d_i^2 + k_i}\,,
\end{equation}
where $k_i = \sigma^2/v_i$ is the component-specific ridge penalty.  When all $k_i = k$, this recovers standard ridge regression.

The connection to SM-EM is that the adaptive prior weights $\hat{\lambda}_j$ in the M-step~\eqref{eq:mstep} play exactly the role of the component-specific penalties $k_i$ in~\eqref{eq:grr}.  SM-EM performs generalized ridge regression at every iteration, with data-adaptive penalties that shrink aggressively in low-signal directions (small $d_i$) and minimally in high-signal directions (large $d_i$).  This anticipates the structure of modern global-local shrinkage priors \citep{carvalho2010horseshoe, polson2012local}; see \citet{polson2019bayesian} for a review of the regularization landscape from Tikhonov to horseshoe.

\citet{goldstein1974ridge} showed that generalized ridge dominates least squares in componentwise MSE when the smallest singular value exceeds a signal-to-noise threshold: $\mathrm{MSE}(\alpha_i^\star) < \mathrm{MSE}(\hat{\alpha}_i)$ for all $i$.  This is a stronger guarantee than James--Stein, which ensures overall MSE reduction but not improvement for each parameter individually.  The SM-EM M-step inherits this componentwise improvement because its adaptive weights $\hat{\lambda}_j$ are derived from the model structure rather than from a global shrinkage factor.

\section{Discussion}\label{sec:discussion}

This paper develops a scale-mixture alternative to Robbins--Monro for losses admitting variance-mean mixture representations.  SM-EM uses latent variables $\omega_i$ and $\lambda_j$ as model-derived curvature and shrinkage weights, eliminating learning-rate, momentum, and decay schedules.  In the synthetic logistic benchmarks of Section~\ref{sec:empirical}, this tuning-free algorithm reaches lower final loss than grid-tuned Adam, with the gap growing in dimension and condition number; sharing sufficient statistics across penalty values further accelerates regularization paths.  These empirical gains are specific to the reported settings (ill-conditioned synthetic logistic regression).

The connection to proximal methods is deep.  \citet{toulis2017asymptotic} showed that implicit SGD incorporates Fisher information to achieve unconditional stability, and \citet{toulis2021proximal} formalized this by establishing that the Robbins--Monro recursion is fundamentally a proximal algorithm---the implicit update $\theta_{t+1} = \theta_t - \alpha_t \nabla \ell(\theta_{t+1})$ is equivalent to applying the proximal operator of $\alpha_t \ell$.  The scale mixture framework goes further: it replaces the global proximal step (which still requires a step size $\alpha_t$) with the envelope representation~\eqref{eq:envelope_scale}, yielding observation-specific curvature estimates that serve as automatic step sizes.  In the language of \citet{toulis2017asymptotic}, the implicit SGD shrinkage factor $(I + \gamma_n\hat{\mathcal{I}}_n)^{-1}$ is replaced by the P\'olya--Gamma precision matrix $(\tau^{-2}\hat{\Lambda} + X^\top\hat{\Omega}X)^{-1}$, which is computed in closed form from the loss structure rather than approximated by a Newton step.  The half-quadratic envelope of \citet{polson2016mixtures} provides the same function as the proximal operator but through profiling rather than integration, connecting MM algorithms (optimization) to EM algorithms (statistical inference) via hierarchical duality.

The framework also provides a principled approach to hyperparameter selection via the evidence (marginal likelihood).  Following \citet{mackay1992bayesian} (and adopting his notation for this paragraph), let $\tau_D^2 = 1/\sigma^2$ and $\tau_w^2 = 1/\sigma_w^2$ denote the data and prior precisions, and let $A = \tau_w^2 C + \tau_D^2 B$ be the Hessian of the negative log-posterior at the MAP estimate, where $B$ and $C$ are the Hessians of the data misfit and prior energy.  The log-evidence under a Laplace approximation is
\[
\log p(D \mid \tau_w^2, \tau_D^2) = -\tau_w^2 E_W^{\mathrm{MP}} - \tau_D^2 E_D^{\mathrm{MP}} - \tfrac{1}{2}\log\det A + \text{const}\,,
\]
and setting its derivative with respect to $\tau_w^2$ to zero yields the effective number of parameters
\begin{equation}\label{eq:eff_dof}
\gamma = k - \tau_w^2\,\mathrm{tr}(A^{-1}) = k - \sum_{a=1}^k \frac{\tau_w^2}{\lambda_a + \tau_w^2}\,,
\end{equation}
where $\{\lambda_a\}$ are the eigenvalues of $B$.  The connection to SM-EM follows from structure: the M-step matrix $A = \tau^{-2}\hat{\Lambda} + X^\top\hat{\Omega}X$ matches MacKay's posterior-Hessian form, with $B = X^\top\hat{\Omega}X$ and $C = \hat{\Lambda}$.  Since the complete sufficient statistics $\{G_i = x_ix_i^\top\}$ are already cached, the trace $\mathrm{tr}(A^{-1})$ can be computed from the Cholesky factor of $A$ at low additional cost, making evidence-based selection of $\tau$ inexpensive within SM-EM iterations.  The same trace computation enables generalized cross-validation \citep{golub1979generalized}, whose score $\mathrm{GCV}(\tau) = n^{-1}\|y - \hat{y}\|^2/(1 - \gamma/n)^2$ depends on $\gamma$ and the residual norm.  The effective degrees of freedom $\gamma$ in~\eqref{eq:eff_dof} also aligns with the shrinkage profiles from Section~\ref{sec:applications}: each eigenvalue $\lambda_a$ contributes $\lambda_a/(\lambda_a + \tau_w^2) = s_a$, the same shrinkage weight as in generalized ridge regression~\eqref{eq:grr}.

Nesterov acceleration complements the framework.  Since each full-batch SM-EM iteration is deterministic, Nesterov extrapolation can be applied without variance reduction---unlike stochastic gradient methods, which require SVRG \citep{johnson2013accelerating} or SAGA \citep{defazio2014saga} to reduce noise before acceleration is effective.  Empirically, SM-EM+Nesterov achieves 32--55\% lower loss on the ill-conditioned benchmarks of Section~\ref{sec:empirical}; the formal $O(1/t^2)$ guarantee for SM-EM's adaptive curvature (as opposed to fixed-step proximal gradient) remains an open question (Appendix~\ref{app:convergence}).

For large-scale problems, the $O(p^3)$ M-step solve can be replaced by Halton's sequential Monte Carlo method \citep{halton1994sequential, polson2026fast}, which achieves geometric convergence with $O(ps)$ per-iteration cost---enabling SM-EM for problems where direct factorization is prohibitive.  The complete sufficient statistics structure of the M-step \citep{polson2011data} further reduces cost: the outer products $x_i x_i^\top$ are precomputed once, and for sparse models the active set shrinks across iterations, making later EM steps progressively cheaper.

Several limitations deserve mention.  First, SM-EM applies only to losses admitting a scale mixture of normals representation; while this covers a broad class (logistic, hinge, check, bridge penalties---see Section~\ref{sec:scalemix}), arbitrary neural network losses do not qualify.  Second, the experiments here use logistic regression on synthetic data; real-data benchmarks and extensions to multi-class or deep models remain future work.  For the multinomial case, \citet{linderman2015dependent} showed that stick-breaking with P\'olya--Gamma augmentation yields conjugate Gaussian updates, suggesting a natural path for extending SM-EM beyond binary classification.  Third, SM-EM as presented is a full-batch method.  For problems where $n$ is too large for full-batch evaluation, a stochastic EM variant is needed; the online EM framework of \citet{cappe2009online} provides a path, though the tuning-free guarantee would require additional analysis.  Fourth, the $O(p^3)$ M-step, while addressed by Halton Monte Carlo for very large $p$, means SM-EM is slower per iteration than SGD-based methods when $p$ is large and the problem is well-conditioned enough for first-order methods to suffice.

Extensions to full posterior inference via MCMC (Gibbs sampling on $\omega, \lambda, \beta$) use the same data augmentation structure, and the Katyusha momentum of \citet{allenzhu2017katyusha} offers a path to accelerated stochastic variants for the finite-sum setting.

\ifarxiv\else
\acks{The authors thank the action editor and reviewers for constructive feedback.}
\fi

\appendix

\section{Convergence Theory for SM-EM with Nesterov Acceleration}\label{app:convergence}

This appendix discusses the convergence theory underlying Nesterov acceleration as applied to the scale mixture EM algorithm.

\subsection{Half-Quadratic Framework}

Consider the composite objective $F(w) = f(w) + \phi(w)$ where $\phi$ is convex.  Suppose $f(w)$ admits a half-quadratic envelope:
\begin{equation}\label{eq:hq_envelope}
f(w) = \inf_v \left\{ f(v) + \nabla f(v)^\top(w - v) + \frac{L(v)}{2}\vnorm{w - v}^2 \right\},
\end{equation}
where the infimum is attained at $v = w$.  The standard proximal gradient method uses a constant $L(v) = L$; in the scalar case, the scale mixture approach uses the adaptive choice $L(v) = (\nabla f(v) + \kappa)/v$ (cf.\ the GR minimizer $\hat{\omega} = \phi'(x)/x$ in Table~\ref{tab:hq_envelopes}), which yields tighter quadratic bounds at points far from the optimum.  In the multivariate case, $L(v)$ is replaced by the observation-specific precision matrix $X^\top\hat{\Omega}X$.  We assume a lower bound $L \preceq L(v)$ for all $v$ to guarantee fast convergence.

Define the joint objective $F(w, v) = f(v) + \nabla f(v)^\top(w - v) + \phi(w)$.  Finding the marginal mode $w^\star = \arg\min_w F(w)$ reduces to alternating minimization of
\[
\min_{w,v}\left\{F(w, v) + \frac{L(v)}{2}\vnorm{w - v}^2\right\}.
\]
The conditional mode for $v$ given $w_t$ is $v_t = w_t$ (the ICM step); the conditional mode for $w$ given $v_t$ is the proximal operator
\begin{equation}\label{eq:prox_step_app}
w_{t+1} = \prox_{L(v_t)^{-1}\phi}\!\left(v_t - L(v_t)^{-1}\nabla f(v_t)\right),
\end{equation}
since completing the square in $w$ yields $\frac{L(v_t)}{2}\vnorm{w - (v_t - L(v_t)^{-1}\nabla f(v_t))}^2 + \phi(w)$ plus terms not depending on $w$.  When $L(v)$ is the adaptive choice from the scale mixture, this proximal step uses the model-derived curvature rather than the global Lipschitz constant.

\subsection{Three-Point Inequality}

For any convex $\phi$ and parameter $\gamma > 0$, the Moreau envelope $E_{\gamma\phi}(x) = \inf_z\{\frac{1}{2\gamma}\enorm{z - x}^2 + \phi(z)\}$ satisfies the descent property
\[
E_{\gamma\phi}(y) \leq E_{\gamma\phi}(x) - \frac{1}{2\gamma}\vnorm{x - y}^2, \quad y = \prox_{\gamma\phi}(x).
\]
When $L \preceq L(x)$ for all $x$, this generalizes to the three-point inequality: for $y = \prox_{L(x)^{-1}\phi}(x)$,
\begin{equation}\label{eq:3point}
\phi(y) + \frac{L(x)}{2}\enorm{y - x}^2 \leq \phi(z) + \frac{L(x)}{2}\enorm{z - x}^2 - \frac{L}{2}\vnorm{z - y}^2 \quad \forall z.
\end{equation}
This reverse-triangle Pythagoras inequality is the key tool for establishing convergence rates.

Combining~\eqref{eq:3point} with the descent lemma $l_f(w, v) \defeq f(w) - f(v) - \nabla f(v)^\top(w - v) \leq \frac{L(v)}{2}\|w - v\|^2$ (which holds when $\nabla f$ is Lipschitz) and convexity ($l_f \geq 0$), each proximal step satisfies
\begin{equation}\label{eq:descent_bound}
F(w_{t+1}) \leq F(w^\star) + \frac{L}{2}\vnorm{w^\star - v_t}^2 - \frac{L}{2}\vnorm{w^\star - w_{t+1}}^2.
\end{equation}

\subsection{$O(1/k^2)$ Convergence with Nesterov Acceleration}

To accelerate convergence beyond the $O(1/k)$ rate of~\eqref{eq:descent_bound}, introduce an auxiliary sequence $u_t$ and weights $\theta_t$ such that
\begin{align}
v_t &= (1 - \theta_t)w_t + \theta_t u_t, \label{eq:vt_def}\\
u_{t+1} &= \frac{1}{\theta_t}\!\left(w_{t+1} - (1 - \theta_t)w_t\right). \label{eq:ut_def}
\end{align}
The intermediate point $v_t$ is a convex combination of $w_t$ and $u_t$, chosen to ensure errors compound at rate $O(1/k^2)$.

\begin{theorem}[Nesterov, 1983; Beck and Teboulle, 2009]\label{thm:nesterov}
Let $f$ be convex with $L$-Lipschitz gradient.  Define $\lambda_0 = 0$, $\lambda_s = (1 + \sqrt{1 + 4\lambda_{s-1}^2})/2$, and $\gamma_s = (1 - \lambda_s)/\lambda_{s+1}$.  Starting from $x_1 = y_1$, iterate
\begin{align*}
y_{s+1} &= x_s - \frac{1}{L}\nabla f(x_s), \\
x_{s+1} &= (1 - \gamma_s)y_{s+1} + \gamma_s y_s.
\end{align*}
Then $f(y_t) - f(x^\star) \leq 2L\|x_1 - x^\star\|^2 / t^2$.
\end{theorem}

\begin{proof}
For any $x, y \in \mathbb{R}^n$, $L$-smoothness of $f$ gives
\begin{align*}
f\!\left(x - \tfrac{1}{L}\nabla f(x)\right) - f(y) &\leq -\tfrac{1}{2L}\|\nabla f(x)\|^2 + \nabla f(x)^\top(x - y) \\
&= -\tfrac{L}{2}\|y_{s+1} - x_s\|^2 - L(y_{s+1} - x_s)^\top(x_s - y).
\end{align*}
Apply with $(x, y) = (x_s, y_s)$ and $(x, y) = (x_s, x^\star)$.  Let $\delta_s = f(y_s) - f(x^\star)$.  Multiply the first inequality by $(\lambda_s - 1)$ and add to the second:
\[
\lambda_s\,\delta_{s+1} - (\lambda_s - 1)\,\delta_s \leq -\tfrac{L}{2}\lambda_s\|y_{s+1} - x_s\|^2 - L(y_{s+1} - x_s)^\top\!\left(\lambda_s x_s - (\lambda_s - 1)y_s - x^\star\right).
\]
Multiplying by $\lambda_s$ and using $\lambda_{s-1}^2 = \lambda_s^2 - \lambda_s$:
\[
\lambda_s^2\,\delta_{s+1} - \lambda_{s-1}^2\,\delta_s \leq \tfrac{L}{2}\!\left(\|u_s\|^2 - \|u_{s+1}\|^2\right),
\]
where $u_s = \lambda_s x_s - (\lambda_s - 1)y_s - x^\star$, using the identity $\lambda_{s+1}x_{s+1} - (\lambda_{s+1} - 1)y_{s+1} = \lambda_s y_{s+1} - (\lambda_s - 1)y_s$ from the definition of $x_{s+1}$.  Telescoping from $s = 1$ to $s = t - 1$:
\[
\lambda_{t-1}^2\,\delta_t \leq \tfrac{L}{2}\|u_1\|^2.
\]
Since $\lambda_{t-1} \geq t/2$ by induction, we obtain $\delta_t \leq 2L\|x_1 - x^\star\|^2/t^2$.
\end{proof}

Theorem~\ref{thm:nesterov} assumes a fixed Lipschitz constant $L$ and uses $1/L$ as the step size.  SM-EM uses adaptive curvature $L(v_t)$ that varies per iteration (Section~\ref{sec:scalemix}), so the theorem does not apply directly.  Extending the $O(1/t^2)$ guarantee to the adaptive setting requires the uniform lower bound $L \preceq L(v)$ from~\eqref{eq:hq_envelope} together with a generalized descent lemma for variable metrics; this extension is a direction for future work.  In the experiments (Figure~\ref{fig:nesterov}), SM-EM+Nesterov exhibits convergence consistent with the $O(1/t^2)$ rate on the ill-conditioned problems tested, and the deterministic full-batch E-step ensures no stochastic noise enters the extrapolation.

\section{Half-Quadratic Envelopes and Proximal Operators}\label{app:envelopes}

\subsection{Table of Half-Quadratic Envelopes}

Table~\ref{tab:hq_envelopes} catalogs the two standard half-quadratic representations for common penalty functions: the Geman--Reynolds (GR) scale representation $\phi(x) = \inf_s\{\frac{1}{2}x^2 s + \psi(s)\}$ with minimizer $\hat{s} = \phi'(x)/x$, and the Geman--Yang (GY) location representation $\phi(x) = \inf_s\{\frac{1}{2}(x - s)^2 + \psi(s)\}$ with minimizer $\hat{s} = \prox_\psi(x) = x - \phi'(x)$.  The GR representation corresponds to the variance mixture (scale mixture of normals); the GY representation corresponds to the location mixture (Moreau envelope).

\begin{table}[h]
\centering
\caption{Half-quadratic envelopes for common penalties.  GR: $Q(x,s) = \frac{1}{2}x^2 s$ (scale/variance mixture); GY: $Q(x,s) = \frac{1}{2}(x - s)^2$ (location/Moreau envelope).}\label{tab:hq_envelopes}
\begin{tabular}{llll}
\hline
\textbf{Penalty} & $\phi(x)$ & \textbf{GR minimizer} $\hat{s}$ & \textbf{GY minimizer} $\hat{s}$ \\
\hline
$L^p$, $\alpha \in (1,2]$ & $|x|^\alpha$ & $\alpha|x|^{\alpha-2}$ & --- \\
$\alpha$-$L^1$ & $\sqrt{\alpha + x^2}$ & $(\alpha + x^2)^{-1/2}$ & $x - x/\sqrt{\alpha + x^2}$ \\
$L^1$-Pareto & $\frac{|x|}{\alpha} - \log(1 + \frac{|x|}{\alpha})$ & $[\alpha(\alpha + |x|)]^{-1}$ & $x - x/[\alpha(\alpha + |x|)]$ \\
Huber & $\begin{cases} x^2/2 & |x| \leq \alpha \\ \alpha|x| - \alpha^2/2 & |x| > \alpha \end{cases}$ & $\min(1, \alpha/|x|)$ & --- \\
Logit & $\log\cosh(\alpha x)$ & $\alpha\tanh(\alpha x)/x$ & $x - \alpha\tanh(\alpha x)$ \\
\hline
\end{tabular}
\end{table}

\subsection{Proximal Operators for $L^p$ Penalties}

The proximal operator $\prox_{\lambda|\cdot|^p}(y) = \arg\min_z\{\frac{1}{2}(y - z)^2 + \lambda|z|^p\}$ has closed-form solutions for several values of $p$:
\begin{align*}
p = 1&: \quad \mathrm{sgn}(y)\max(|y| - \lambda, 0) &&\text{(soft thresholding)}\\
p = \tfrac{3}{2}&: \quad y + \tfrac{9\lambda^2}{8}\,\mathrm{sgn}(y)\!\left(1 - \sqrt{1 + \tfrac{16|y|}{9\lambda^2}}\right) \\
p = 2&: \quad y/(1 + 2\lambda) &&\text{(ridge)}\\
p = 3&: \quad \mathrm{sgn}(y)\!\left(\sqrt{1 + 12\lambda|y|} - 1\right)\!/(6\lambda)
\end{align*}
For $0 < p < 1$ (concave penalties), a threshold $y^\star = (\lambda p(1-p))^{1/(2-p)}$ separates the shrinkage-to-zero region from the continuous shrinkage region, yielding sparse solutions.

\subsection{Lasso via the AM-GM Identity}

The $L^1$ norm admits the quadratic representation
\begin{equation}\label{eq:amgm}
\vnorm{\beta}_1 = \sum_{i=1}^p \inf_{\lambda_i \geq 0}\left\{\frac{\beta_i^2}{2\lambda_i} + \frac{\lambda_i}{2}\right\},
\end{equation}
with equality at $\lambda_i = |\beta_i|$, via the arithmetic-geometric mean inequality $\sqrt{ab} \leq (a + b)/2$.  This gives a joint convex objective
\[
\inf_{\beta, \lambda}\left\{\frac{1}{2}\vnorm{y - X\beta}^2 + \rho\sum_{i=1}^p\left(\frac{\beta_i^2}{\lambda_i} + \lambda_i\right)\right\},
\]
which reduces to iteratively reweighted ridge regression: the $\lambda$-step sets $\hat{\lambda}_i = |\beta_i|$, and the $\beta$-step solves ridge with adaptive penalties $\rho/\hat{\lambda}_i$.  This is precisely the variance-mean mixture representation of the Lasso, connecting the $L^1$ proximal operator (soft thresholding) to the scale mixture EM framework.

\subsection{Iterative Shrinkage as EM}

The iterative shrinkage-thresholding (IST) algorithm for $\min_\beta \frac{1}{2}\|y - X\beta\|^2 + \tau\phi(\beta)$ can be derived as an EM algorithm with the data augmentation
\begin{align*}
p(y \mid z) &\sim N(y \mid Xz,\; I - \alpha^{-1}XX^\top), \\
p(z \mid \beta) &\sim N(z \mid \beta,\; \alpha^{-1}I),
\end{align*}
where $\alpha$ is the step-length parameter.  The E-step computes $\hat{z}^{(k)} = \beta^{(k)} + \alpha^{-1}X^\top(y - X\beta^{(k)})$ (a gradient step), and the M-step applies the proximal operator $\beta^{(k+1)} = \prox_{\tau\phi/\alpha}(\hat{z}^{(k)})$ (a shrinkage step).  Substituting recovers the standard IST iteration $\beta^{(k+1)} = \prox_{\tau\phi/\alpha}(\beta^{(k)} - \alpha^{-1}X^\top(X\beta^{(k)} - y))$.  The EM interpretation guarantees monotonic decrease of the penalized likelihood at each iteration, and connects IST to the broader scale mixture framework.

\section{Proximal Point Algorithms, Quasi-Newton, and Bregman Duality}\label{app:ppa}

\subsection{EM and MM as Proximal Point Algorithms}

The proximal point algorithm (PPA) of \citet{rockafellar1976monotone} seeks the minimum of $J(x)$ via
\begin{equation}\label{eq:ppa}
x^{(k+1)} = \arg\min_x\left\{J(x) + \eta^{(k)} D(x, x^{(k)})\right\},
\end{equation}
where $D(x, y) \geq 0$ with $D(x, x) = 0$ is a proximity measure (originally $D = \frac{1}{2}\|x - y\|^2$, but Bregman divergences, Kullback--Leibler, and other choices are possible).  The construction guarantees monotonicity:
\[
J(x^{(k+1)}) \leq J(x^{(k)})\,.
\]
The step size $\eta^{(k)}$ controls the conservatism of each update.

Both EM and MM algorithms are special cases of PPA\@.  The EM algorithm chooses $D$ from the complete-data log-likelihood and fixes $\eta^{(k)} = 1$.  This unit step size is part of why EM converges slowly---there is no mechanism to take larger steps.  The MM (majorization--minimization) algorithm constructs a majorizing function $Q(x \mid \tilde{x}) \geq J(x)$ with $Q(x \mid x) = J(x)$; the iterate $x^{(k+1)} = \arg\min\,Q(\cdot \mid x^{(k)})$ then satisfies $J(x^{(k+1)}) \leq J(x^{(k)})$, again a PPA with a specific majorizing envelope.

The SM-EM algorithm inherits this PPA structure, but the half-quadratic envelope provides a tighter majorizer than the generic EM bound: the adaptive weights $\hat{\omega}_i$ shape the quadratic proximity measure to the local curvature of the loss, yielding an effective step size larger than the unit step of standard EM while retaining the monotonic descent guarantee.  Nesterov-style algorithms sacrifice monotonicity for faster $O(1/k^2)$ convergence---the iterates spiral toward the solution.  When Nesterov extrapolation is applied to SM-EM (Appendix~\ref{app:convergence}), the outer monotonicity guarantee is lost, but each inner proximal step still uses the adaptive half-quadratic envelope; empirically, the resulting iterates converge at a rate consistent with $O(1/k^2)$, though the formal guarantee for adaptive curvature remains open.

\subsection{Half-Quadratic Penalties as Quasi-Newton}\label{app:quasinewton}

The half-quadratic iteration has a direct interpretation as a quasi-Newton method, deepening the preconditioning analogy from Section~\ref{sec:empirical}.  Consider the cost functional
\[
J(x) = \frac{1}{2}\|Ax - y\|^2 + \lambda\,\phi(x)\,,\qquad \phi(x) = \sum_{i=1}^p \phi\!\left((G^\top x - w)_i\right),
\]
where $G$ encodes a linear operator (e.g., the identity for direct penalization).  Introducing dual variables $(\lambda_1, \ldots, \lambda_p)$ such that $J(x) = \inf_\lambda J(x, \lambda)$, the joint objective is quadratic in $x$ given $\lambda$.

If $\phi$ is differentiable, the gradient is
\begin{align*}
\nabla_x J(x) &= A^\top A\,x - A^\top y + \lambda\sum_{i=1}^p G_i\,\frac{\phi'(\|\delta_i\|)}{\|\delta_i\|}\,G_i^\top x \\
&= \bigl(A^\top A + \lambda\,G\,\Lambda(x)\,G^\top\bigr)x - A^\top y = L(\hat{\lambda}(x))\,x - A^\top y\,,
\end{align*}
where $\delta_i = (G^\top x - w)_i$, $\Lambda(x) = \mathrm{diag}(\phi'(\|\delta_i\|)/\|\delta_i\|)$, and $L(\hat{\lambda}(x)) = A^\top A + \lambda\,G\,\Lambda(x)\,G^\top$.

Gradient linearization gives the iterative mapping $x^{(k+1)} = L(\hat{\lambda}(x^{(k)}))^{-1}A^\top y$, while quasi-Newton computes $x^{(k+1)} = x^{(k)} - L(x^{(k)})^{-1}\nabla_x J(x^{(k)})$.  These are the same iteration: the half-quadratic update is a quasi-Newton step with the normal equation matrix $L(\hat{\lambda}(x))$ serving as the Hessian approximation.

For the Geman--Reynolds (scale) representation, the normal equation matrix is $B_{\mathrm{GR}} = A^\top A + G\,\Lambda(x)\,G^\top$ with $\hat{\Lambda} = \mathrm{diag}(\phi'(\delta_i)/\delta_i)$, and the iteration is
\begin{align*}
\hat{\lambda}_i^{(k+1)} &= \phi'(\delta_i^{(k)})/\delta_i^{(k)}\,, \\
x^{(k+1)} &= B_{\mathrm{GR}}^{-1}\bigl(A^\top y + G\,\hat{\Lambda}^{(k+1)}w\bigr)\,.
\end{align*}
For the Geman--Yang (location) representation with rescaling parameter $a = 1/L$ (the reciprocal Lipschitz constant of $\phi'$), the matrix is $B_{\mathrm{GY}}^a = A^\top A + a^{-1}GG^\top$, which does not depend on the iterate---it is a fixed preconditioner computed once.  The iteration becomes
\begin{align*}
\hat{\lambda}_i^{(k+1)} &= \delta_i^{(k)} - a\,\phi'(\delta_i^{(k)})\,, \\
x^{(k+1)} &= (B_{\mathrm{GY}}^a)^{-1}\bigl(A^\top y + a^{-1}G(\hat{\lambda}^{(k+1)} + w)\bigr)\,.
\end{align*}

The connection to preconditioning is now explicit.  The Geman--Reynolds (scale mixture) iteration uses an adaptive preconditioner $B_{\mathrm{GR}}$ that changes at each step---exactly the mechanism exploited by SM-EM, where the precomputed outer products $\{G_i = x_i x_i^\top\}$ are the offline building blocks and the weights $\hat{\omega}_i$ provide the adaptive rescaling.  The Geman--Yang iteration uses a fixed preconditioner $B_{\mathrm{GY}}^a$, analogous to a one-time incomplete Cholesky factorization.  The variance-mean mixture, which generalizes both, interpolates: the scale component provides adaptive curvature, while the location component provides a fixed stabilizing term.

\subsection{Formal Derivation of Scale and Location Mixtures}

\begin{theorem}[Location mixtures; Geman--Yang]\label{thm:gy}
Suppose that $\theta(x) = \frac{1}{2}x^2 - \phi(x)$ is a convex function.  Then
\begin{align*}
\phi(x) &= \inf_{\lambda \in \mathcal{R}}\left\{\frac{1}{2}(x - \lambda)^2 + \psi(\lambda)\right\}, \\
\psi(\lambda) &= \frac{1}{2}\lambda^2 - \theta^\star(\lambda)\,,
\end{align*}
where $\theta^\star$ is the convex conjugate of $\theta$.  If $\phi$ is differentiable, the infimum is attained at $\hat{\lambda} = \prox_\psi(x) = x - \phi'(x)$.
\end{theorem}
\begin{proof}
The convex conjugate of $\theta(x) = \frac{1}{2}x^2 - \phi(x)$ satisfies
\[
\theta(x) = \sup_{\lambda}\left\{\lambda\,x - \theta^\star(\lambda)\right\} = \sup_\lambda\left\{\lambda\,x - \frac{1}{2}\lambda^2 + \psi(\lambda)\right\}.
\]
Hence $\phi(x) = \frac{1}{2}x^2 - \theta(x) = \inf_\lambda\{\frac{1}{2}x^2 - \lambda\,x + \frac{1}{2}\lambda^2 - \psi(\lambda)\} = \inf_\lambda\{\frac{1}{2}(x - \lambda)^2 + \psi(\lambda)\}$.  When $\phi$ is differentiable, the first-order condition $\hat{\lambda} - x + \psi'(\hat{\lambda}) = 0$ gives $\hat{\lambda} = x - \phi'(x)$, which is $\prox_\psi(x)$ by definition.
\end{proof}

For the scale (Geman--Reynolds) representation, if $\phi(\sqrt{x})$ is convex then $\theta(x) = -\phi(\sqrt{x})$ is concave and the concavity inequality gives
\[
\phi(x) = \inf_{\omega \geq 0}\left\{\omega\,x^2 + \psi(\omega)\right\}, \quad \hat{\omega}(x) = \frac{\phi'(x)}{2x}\,.
\]

The variance-mean mixture generalizes both.  Assume $\theta(x) = -(\phi(\sqrt{x}) + 2\kappa\sqrt{x})$ is concave.  Since $\theta'(x)$ is strictly decreasing, the concavity inequality $\theta(x) \leq \theta'(y)(x - y) + \theta(y)$ with equality at $y = x$ yields
\begin{align*}
\phi(\sqrt{x}) &= \inf_\omega\left\{\omega\,x - 2\kappa\sqrt{x} + \psi(\omega)\right\} \\
&= \inf_\omega\left\{\omega(\sqrt{x} - \kappa\,\omega^{-1})^2 + \psi(\omega) - \kappa^2\omega^{-1}\right\},
\end{align*}
where $\hat{\omega}(x) = \theta'(x) = (\phi'(x) + 2\kappa)/(2x)$.  The transformation $\sqrt{x} \to x$ recovers the variance-mean mixture~\eqref{eq:vmm} used in the main text.  The parameter $\kappa$ handles asymmetric losses: setting $\kappa = 0$ in the scale form recovers Geman--Reynolds, and setting $\kappa = 0$ in the location form recovers Geman--Yang.

\subsection{Bregman Divergence and Exponential Families}\label{app:bregman}

The Bregman divergence of a differentiable strictly convex function $\phi$ measures ``how convex'' $\phi$ is:
\[
D_\phi(x, y) = \phi(x) - \phi(y) - \nabla\phi(y)^\top(x - y) \geq 0\,,
\]
with equality if and only if $x = y$.  Common choices include: $\phi(z) = \frac{1}{2}z^\top z$ giving $D_\phi = \frac{1}{2}\|x - y\|^2$ (squared Euclidean); $\phi(z) = z\log z$ giving $D_\phi = x\log(x/y) - x + y$ (Kullback--Leibler); $\phi(z) = -\log z$ giving $D_\phi = x/y - \log(x/y) - 1$ (Itakura--Saito).

The connection to our framework is through natural exponential families.  If $\psi(\theta)$ is the cumulant generating function and $\phi$ is its Legendre dual, then $p_\psi(x \mid \theta) = \exp(-D_\phi(x, \mu(\theta)) - h_\phi(x))$, where $\mu(\theta) = \psi'(\theta)$ is the mean parameter.  MAP estimation becomes Bregman divergence minimization:
\[
\max_\theta\,\sum_{i=1}^n \log p_\psi(x_i \mid \theta) \quad\Leftrightarrow\quad \min_\mu\,\sum_{i=1}^n D_\phi(x_i, \mu(\theta))\,.
\]
This provides a direct link between the SM-EM framework and natural gradient methods \citep{amari1998natural}: the Fisher information matrix $I(\theta) = \nabla^2\psi(\theta)$ is the Hessian of the Bregman generator, so the natural gradient step $\theta_{t+1} = \theta_t - I(\theta_t)^{-1}\nabla\ell(\theta_t)$ is a proximal step in the Bregman geometry---the same structure as the SM-EM M-step, but with the Bregman divergence replacing the Euclidean proximity.  The PPA framework of~\eqref{eq:ppa} unifies both: choosing $D = D_\phi$ (Bregman) gives mirror descent and natural gradient; choosing $D$ via the half-quadratic envelope gives SM-EM\@.

\bibliography{ref}
\end{document}